%% file: stochastic-rm.tex
\documentclass[letterpaper]{article}

\usepackage[nolinenums]{custom_arxiv}
\usepackage[utf8]{inputenc} 
\usepackage[T1]{fontenc}    
\usepackage{hyperref}       
\usepackage{url}            
\usepackage{booktabs}       
\usepackage{amsfonts}       
\usepackage{nicefrac}       
\usepackage{microtype}      
\usepackage{times}
\usepackage{amsmath}
\usepackage{amssymb}
\usepackage{amsthm}
\usepackage{microtype}
\usepackage{enumitem}
\usepackage{dsfont}
\usepackage[capitalise,noabbrev]{cleveref}
\usepackage{xcolor}
\usepackage{etoolbox}
\usepackage{thm-restate}
\usepackage[linesnumbered,ruled,lined,noend]{algorithm2e}
\usepackage{amsmath}
\usepackage{mathtools}
\usepackage{amsthm}
\usepackage{mleftright}
\usepackage{stmaryrd}
\usepackage{natbib}
\usepackage{float}
\usepackage{bm}
\usepackage{numprint}
\usepackage{tikz}
\usepackage{cleveref}
\usepackage{standalone}

\usepgflibrary{shapes.geometric,arrows.meta}

\ifcsmacro{assumption}{}{
	
}
\ifcsmacro{corollary}{}{
	
}
\ifcsmacro{definition}{}{
	
}
\ifcsmacro{example}{}{
	
}
\ifcsmacro{fact}{}{
	
}
\ifcsmacro{informal_theorem}{}{
	
}
\ifcsmacro{lemma}{}{
	\newtheorem{lemma}{Lemma}[]
}
\ifcsmacro{proposition}{}{
	\newtheorem{proposition}{Proposition}[]
}
\ifcsmacro{remark}{}{
	
}
\ifcsmacro{theorem}{}{
	\newtheorem{theorem}{Theorem}[]
}
\ifcsmacro{observation}{}{
	
}
\mathtoolsset{centercolon}

\makeatletter
\patchcmd\algocf@Vline{\vrule}{\vrule \kern-0.4pt}{}{}
\patchcmd\algocf@Vsline{\vrule}{\vrule \kern-0.4pt}{}{}
\makeatother

\SetKwComment{Hline}{}{\vspace{-3mm}\textcolor{gray}{\hrule}\vspace{1mm}}
\definecolor{darkgrey}{gray}{0.3}
\definecolor{commentcolor}{gray}{0.5}
\SetKwComment{Comment}{\color{commentcolor}$\triangleright$\ }{}
\SetCommentSty{}
\SetNlSty{}{\color{darkgrey}}{}
\crefname{algocf}{Algorithm}{Algorithms}
\SetKwProg{Fn}{function}{}{}
\SetKwProg{Subr}{subroutine}{}{}

\newcommand{\defeq}{\mathrel{:\mkern-0.25mu=}}

\newcommand{\cX}{\mathcal{X}}
\newcommand{\cY}{\mathcal{Y}}

\newcommand{\cZ}{\mathcal{Z}}

\newcommand{\bbR}{\mathbb{R}}
\newcommand{\bbP}{\mathbb{P}}
\newcommand{\bbE}{\mathbb{E}}
\newcommand{\R}{\mathbb{R}}
\newcommand{\dmid}{\parallel}

\renewcommand{\vec}[1]{\bm{#1}}
\newcommand{\mat}[1]{\bm{#1}}

\DeclareMathOperator{\var}{Var}
\DeclareMathOperator*{\argmax}{arg\,max}
\DeclareMathOperator*{\argmin}{arg\,min}

\newcommand{\chancepl}{\textsf{c}}
\newcommand{\infos}[1]{\mathcal{I}_{#1}}
\newcommand{\emptyseq}{\varnothing}

\title{\fontsize{16}{16}\selectfont Stochastic Regret Minimization in Extensive-Form Games}
\author{Gabriele Farina\\
Computer Science Department\\
Carnegie Mellon University\\
\texttt{gfarina@cs.cmu.edu}
\And Christian Kroer\\
IEOR Department\\
Columbia University\\
\texttt{christian.kroer@columbia.edu}
\And Tuomas Sandholm\\
Computer Science Department, CMU\\
Strategic Machine, Inc.\\
Strategy Robot, Inc.\\
Optimized Markets, Inc.\\
\texttt{sandholm@cs.cmu.edu}
}

\DeclareMathSizes{10}{9}{7}{5}
\makeatletter
\g@addto@macro \normalsize {%
 \addtolength\abovedisplayskip{-3pt}%
 \addtolength\belowdisplayskip{-3pt}%
}
\makeatother

\begin{document}
    \twocolumn[
        \maketitle
    ]

    \begin{abstract}
        Monte-Carlo counterfactual regret minimization (MCCFR) is the
        state-of-the-art algorithm for solving sequential games that are
        too large for full tree traversals. It works by using gradient
        estimates that can be computed via sampling. However, stochastic
        methods for sequential games have not been investigated extensively
        beyond MCCFR. In this paper we develop a new framework for
        developing stochastic regret minimization methods. This framework
        allows us to use any regret-minimization algorithm, coupled with
        any gradient estimator. The MCCFR algorithm can be analyzed as a
        special case of our framework, and this analysis leads to
        significantly-stronger theoretical guarantees on convergence, while
        simultaneously yielding a simplified proof. Our framework allows us
        to instantiate several new stochastic methods for solving
        sequential games. We show extensive experiments on three games,
        where some variants of our methods outperform MCCFR.
    \end{abstract}

\input{text/introduction}
    \input{text/preliminaries}
    \input{text/stochastic}
    \input{text/estimators}
    \input{text/experiments}
    \input{text/conclusion}

    \bibliographystyle{custom_arxiv}
    \bibliography{dairefs}

\iftrue
    \clearpage
    \onecolumn
    \appendix

    \input{text/appendix_proofs}
    \input{text/appendix_games}
    \input{text/appendix_experiments}
\fi
\end{document}

%% file: text/introduction.tex
\section{Introduction}
\label{sec:intro}
Extensive-form games (EFGs) are a broad class of games that can model sequential
and simultaneous moves, outcome uncertainty, and imperfect information. This
includes real-world settings such as negotiation, sequential auctions, security
games~\citep{Lisy16:Counterfactual,Munoz13:Introducing}, cybersecurity
games~\citep{Debruhl14:Power,Chen18:Wireless}, recreational games such as
poker~\citep{Sandholm10:State} and billiards~\citep{Archibald09:Modeling}, and
medical treatment~\citep{Chen12:Tractable}. Typically, EFG
models are operationalized by computing either a \emph{Nash equilibrium} of the game,
or an approximate Nash equilibrium if the game is large. Approximate Nash equilibrium of zero-sum
EFGs has been the underlying idea of several recent AI milestones, where
superhuman AIs for two-player poker were
created~\citep{Moravvcik17:DeepStack,Brown17:Superhuman}. In principle, a
zero-sum EFG can be solved in polynomial time using a linear program whose size
is linear in the size of the game tree~\citep{Stengel96:Efficient}. However, for
most real-world games this linear program is much too large to solve, either
because it does not fit in memory, or because iterations of the simplex
algorithm or interior-point methods become prohibitively expensive due to matrix
inversion. Instead, gradient-based methods are used in
practice~\citep{Zinkevich07:Regret,Hoda10:Smoothing,Tammelin15:Solving,Brown19:Solving,Kroer20:Faster}.
These methods work by only keeping one or two strategies around for each player
(typically the size of a strategy is much smaller than the size of the game
tree). The game tree is then only accessed for computing gradients, which can be done via a single tree traversal that does not require an explicit tree representation, and sometimes game structure can be exploited to speed this up further~\citep{Johanson11:Accelerating}. Finally, these gradients are used to update the strategy iterates.

However, eventually even these gradient-based methods that require traversing the entire game tree become too expensive (henceforth referred to as deterministic methods). This was seen in two recent superhuman poker AIs: Libratus~\citep{Brown17:Superhuman} and Pluribus~\citep{Brown19:Superhuman}. Both AIs were generated in a two-stage  manner: an offline blueprint strategy was computed, and then refinements to the blueprint solution were computed online while playing against the human players. The online solutions were computed using deterministic methods (as the subgames are much smaller). However, the original blueprint strategies had to be computed without traversing the entire game tree, as this game tree is far too large for even a moderate amount of traversals.

When full tree traversals are too expensive, stochastic methods can be used to compute approximate gradients instead. The most popular stochastic method for solving EFGs is the \emph{Monte-Carlo Counterfactual Regret Minimization} (MCCFR) algorithm~\citep{Lanctot09:Monte}. MCCFR combines the CFR algorithm~\citep{Zinkevich07:Regret} with various stochastic gradient estimators. Variations of this algorithm were used in both Libratus and Pluribus for computing the blueprint strategy. Follow-up papers have been written on MCCFR, investigating various methods for improving the sampling schemes used in estimating gradients and so on~\citep{Gibson12:Generalized,Schmid19:Variance}. However, beyond the MCCFR setting, stochastic methods have not been studied extensively for solving EFGs\footnote{\citet{Kroer15:Faster} studies the stochastic mirror prox algorithm for EFGs, but it is not the primary focus of the paper, and seems to be more of a preliminary investigation.}.

In this paper we develop a general framework for developing stochastic regret-minimization methods for solving EFGs. In particular, we introduce a way to combine \emph{any} regret-minimizing algorithm with \emph{any} gradient estimator, and obtain high-probability bounds on the performance of the resulting combined algorithm.
As a first application of our approach, we show that with probability $1-p$, the regret in MCCFR is at most $O(\sqrt{\log(1/p)})$ worse off than that of CFR, an exponential improvement over the bound $O(\sqrt{1/p})$ currently known in the literature. Secondly, our approach enables us to develop a slew of other stochastic methods for solving EFGs. As an example of this, we instantiate \emph{follow-the-regularized-leader} (FTRL) and \emph{online mirror descent} (OMD) using our framework.
We then provide extensive numerical simulations on four diverse games, showing that it is possible to beat MCCFR in several of the games using our new methods.
Because of the flexibility and modularity of our approach, it paves the way for many potential future investigations into stochastic methods for EFGs, either via better gradient estimators, or via better deterministic regret minimization methods that can now be converted into stochastic methods, or both.


%% file: text/preliminaries.tex
\section{Preliminaries}

\subsection{Two-Player Zero-Sum Extensive-Form Games}\label{sec:efgs}
In this subsection we introduce the notation that we will use in the rest the paper when dealing with two-player zero-sum extensive-form games.

An extensive-form game is played on a tree rooted at a node $r$. Each node $v$ in the tree belongs to a player from the set $\{1, 2, \chancepl\}$, where $\chancepl$ is called the \emph{chance player}. The chance player plays actions from a fixed distribution known to Player $1$ and $2$, and it is used as a device to model stochastic events such as drawing a random card from a deck. We denote the set of actions available at node $v$ with $A_v$. Each action corresponds to an outgoing edges from $v$. Given $a \in A_v$, we let $\rho(v, a)$ denote the node that is reached by following the edge corresponding to action $a$ at node $v$. Nodes $v$ such that $A_v = \emptyset$ are called \emph{leaves} and represent terminal states of the game. We denote with $Z$ the set of leaves of the game. Associated with each leaf $z \in Z$ is a pair $(u_1(z), u_2(z)) \in \bbR^2$ of payoffs for Player 1 and 2, respectively. We denote with $\Delta$ the \emph{payoff range} of the game, that is the value $\Delta \defeq \max_{z\in Z}\max\{u_1(z), u_2(z)\} - \min_{z\in Z}\min\{u_1(z), u_2(z)\}$. In this paper we are concerned with \emph{zero-sum} extensive-form games, that is games in which $u_1(z) = -u_2(z)$ for all $z\in Z$.

To model private information, the set of all nodes for each player $i \in \{1, 2, \chancepl\}$ is partitioned into a collection $\infos{i}$ of non-empty sets, called \emph{information sets}. Each information set $I \in \infos{i}$ contains nodes that Player $i$ cannot distinguish among. In this paper, we will only consider $\emph{perfect-recall}$ games, that is games in which no player forgets what he or she observed or knew earlier. Necessarily, if two nodes $u$ and $v$ belong to the same information set $I$, the set of actions $A_u$ and $A_v$ must be the same (or the player would be able to tell $u$ and $v$ apart). So, we denote with $A_I$ the set of actions of any node in $I$.

\textbf{Sequences.}\quad
The set of \emph{sequences} for Player $i$, denoted $\Sigma_i$, is defined as the set of all possible information set-action pairs, plus a special element called \emph{empty sequence} and denoted $\emptyseq$. Formally, $\Sigma_i \defeq \{(I, a): I \in \infos{i}, a \in A_{I}\} \cup \{\emptyseq\}$. Given a node $v$ for Player $i$, we denote with $\sigma_i(I)$ the last information set-action pair of Player $i$ encountered on the path from the root to node $v$; if the player does not act before $v$, $\sigma_i(I) = \emptyseq$. It is known that in perfect-recall games $\sigma_i(u) = \sigma_i(v)$ for any two nodes $u,v$ in the same information set. For this reason, for each information set $I$ we define $\sigma_i(I)$ to equal $\sigma_i(v)$ for any $v \in I$.

\textbf{Sequence-Form Strategies.}\quad
A {strategy} for Player $i \in \{1, 2, \chancepl\}$ is an assignment of a probability distribution over the set of actions $A_I$ to each information set $I$ that belongs to Player $i$. In this paper, we represent strategies using their \emph{sequence-form representation}~\citep{Romanovskii62:Reduction,Koller96:Efficient,Stengel96:Efficient}.
A \emph{sequence-form strategy} for Player $i$ is a non-negative vector $\vec{z}$ indexed over the set of sequences $\Sigma_i$ of that player. For each $\sigma = (I, a) \in \Sigma_i$, the entry $z[\sigma]$  contains the product of the probability of all the actions that Player $i$ takes on the path from the root of the game tree down to action $a$ at information set $I$, included. In order for these probabilities to be consistent, it is necessary and sufficient that $z[\emptyseq] = 1$ and
\[
    \sum_{a \in A_I} z[(I, a)] = z[\sigma_i(I)]\quad \forall I\in\infos{i}.
\]
A strategy such that exactly one action is selected with probability $1$ at each node is called a \emph{pure} strategy.

We denote with $\cX$ and $\cY$ the set of all sequence-form strategies for Player $1$ and Player $2$, respectively. We denote with $\vec{c}$ the fixed sequence-form strategy of the chance player.

For any leaf $z\in Z$, the probability that the game ends in $z$ is the product of the probabilities of all the actions on the path from the root to $z$. Because of the definition of sequence-form strategies, when Player 1 and 2 play according to strategies $\vec{x} \in \cX$ and $\vec{y}\in\cY$, respectively, this probability is equal to $x[\sigma_1(z)]\cdot y[\sigma_2(z)] \cdot c[\sigma_\chancepl(z)]$. So, Player 2's expected utility is computed via the trilinear map
\begin{equation}\label{eq:expected payoff pl2}
 \bar{u}_2(\vec{x}, \vec{y}, \vec{c}) \defeq \sum_{z\in Z} u_2(z) \cdot x[\sigma_1(z)]\cdot y[\sigma_2(z)] \cdot c[\sigma_\chancepl(z)].
\end{equation}
Since the strategy of the chance player is fixed, the above expression is bilinear in $\vec{x}$ and $\vec{y}$ and therefore can be expressed more concisely as $\bar{u}_2(\vec{x}, \vec{y}) = \vec{x}^{\!\top}\! \mat{A}_2\, \vec{y}$, where $\mat{A}_2$ is called the \emph{sequence-form payoff matrix} of Player 2.

\subsection{Regret Minimization}
In this section we present the regret minimization algorithms that we will work with. We will operate within the framework of \emph{online convex optimization}~\citep{Zinkevich03:Online}. In this setting, a decision maker repeatedly makes decisions $\vec{z}^1,\vec{z}^2,\ldots$ from some convex compact set $\cZ \subseteq \R^n$. After each decision $\vec{z}^t$ at time $t$, the decision maker faces a \emph{linear loss} $\vec{z}^t \mapsto (\vec{\ell}^t)^{\!\top}\!\vec{z}^t$, where $\vec{\ell}^t$ is a \emph{gradient vector} in $\R^n$. 

Give $\hat{\vec{z}}\in\cZ$, the \emph{regret compared to $\vec{z}$} of the regret minimizer up to time $T$, denoted as $R^T(\hat{\vec{z}})$, measures the difference between the loss cumulated by the sequence
of output decisions $\vec{z}^1, \dots, \vec{z}^T$ and the loss that would have been cumulated by playing a fixed,
time-independent decision $\hat{\vec{z}} \in \cZ$. In symbols,
$
  R^T(\hat{\vec{z}}) \defeq \sum_{t=1}^T (\vec{\ell}^t)^{\!\top}\!(\vec{z}^t - \hat{\vec{z}}).
$
A \emph{Hannan-consistent} regret minimizer is such that the regret compared to \emph{any} $\hat{\vec{z}} \in \cZ$ grows \emph{sublinearly in $T$}.


The two algorithms beyond MCCFR that we consider assume access to a \emph{distance-generating function} $d: \cZ \rightarrow \R$, which is $1$-strongly convex (wrt. some norm) and continuously differentiable on the interior of $\cZ$. Furthermore $d$ should be such that the gradient of the convex conjugate $\nabla d(\vec{g}) = \argmax_{\vec{z} \in \cZ} \langle \vec{g}, \vec{z} \rangle - d(\vec{z})$ is easy to compute. 
From $d$ we also construct the \emph{Bregman divergence}
  $
    D(\vec{z}\dmid \vec{z}') \defeq d(\vec{z}) - d(\vec{z}') - \langle \nabla d(\vec{z}'), \vec{z} - \vec{z}' \rangle.
  $

We will use the following two classical regret minimization algorithms.
The \emph{online mirror descent} (OMD) algorithm produces iterates according to the rule
\begin{align}
  \label{eq:omd}
\vec{z}^{t+1} = \argmin_{\vec{z} \in \cZ} \bigg\{ \langle \vec{\ell}^{t}, \vec{z} \rangle + \frac{1}{\eta} D(\vec{z} \dmid \vec{z}^{t})  \bigg\}.
\end{align}

The \emph{follow the regularized leader} (FTRL) algorithm produces iterates according to the rule~\citep{Schwartz07:Primal}
\begin{align}
  \label{eq:ftrl}
\vec{z}^{t+1} = \argmin_{\vec{z} \in \cZ} \bigg\{ \bigg\langle \sum_{\tau=1}^{t} \vec{\ell}^\tau, \vec{z}\bigg\rangle + \frac{1}{\eta} d(\vec{z})  \bigg\}.
\end{align}

OMD and FTRL satisfy regret bounds of the form
\[\max_{\hat{\vec{z}}\in\cZ} R^T(\hat{\vec{z}}) \leq 2 L \sqrt{D(\vec{z}^*\|
  \vec{z}^1)T},\] where $L$ is an upper bound on $\max_{\vec{x} \in \bbR^n} \frac{(\vec{\ell}^t)^{\!\top}\!\vec{x}}{\|
  \vec{x}\|}$ for all $t$. Here $\| \cdot \|$ is the norm which
we measure strong convexity of $d$ with respect to. (see, e.g.,
\citet{Orabona19:Modern}).

\subsection{Equilibrium-Finding in Extensive-Form Games using Regret Minimization}\label{sec:blsp}
It is known that in a two-player extensive-form game, a Nash equilibrium (NE) is the solution to the \emph{bilinear saddle-point problem}
\[
  \min_{\hat{\vec{x}} \in \cX} \max_{\hat{\vec{y}}\in\cY} \hat{\vec{x}}^{\!\top}\!\! \mat{A}_2\, \hat{\vec{y}}.
\]
Given a pair $(\vec{x}, \vec{y}) \in \cX\times\cY$ of sequence-form strategies for the Player $1$ and $2$, respectively, the \emph{saddle-point gap}
\[
    \xi(\vec{x}, \vec{y}) \defeq \max_{\hat{\vec{y}}\in\cY} \{{\vec{x}}^{\!\top}\!\! \mat{A}_2\, \hat{\vec{y}}\} - \min_{\hat{\vec{x}}\in\cX} \{\hat{\vec{x}}^{\!\top}\!\! \mat{A}_2\, {\vec{y}}\}
\]
measures of how far the pair is to being a Nash equilibrium.
In particular, $(\vec{x}, \vec{y})$ is a Nash equilibrium if and only if $\xi(\vec{x},\vec{y}) = 0$.

Regret minimizers can be used to find a sequence of points $(\vec{x}^t, \vec{y}^t)$ whose saddle-point gap converges to 0. The fundamental idea is to instantiate two regret minimizers $\mathcal{R}_1$ and $\mathcal{R}_2$ for the sets $\cX$ and $\cY$, respectively, and let them respond to each other in a self-play fashion using a particular choice of loss vectors (see \cref{fig:selfplay}).

\begin{figure}[ht]
  \centering
  \begin{tikzpicture}[scale=0.8]
    \draw[semithick] (0, 0) rectangle (1.2, .8);
    \node at (.6, .4) {$\mathcal{R}_1$};
    \draw[semithick] (0, -1) rectangle (1.2, -0.2);
    \node at (.6, -.6) {$\mathcal{R}_2$};
    \draw[->] (-.8, .4) -- (0, .4) node[above left] {$\vec{\ell}_1^{t-1}$};
    \draw[->] (-.8, -.6) -- (0, -.6) node[above left] {$\vec{\ell}_2^{t-1}$};
    \draw[->] (1.2, .4) node[above right] {$\vec{x}^t$} -- (2.0, .4);
    \draw[->] (1.2, -.6) node[above right] {$\vec{y}^t$} -- (2.0, -.6);

    \draw[semithick] (2.0, .2) rectangle (2.4, .6);
    \draw[semithick] (2.0, -.8) rectangle (2.4, -.4);

    \draw[->] (2.4, .4) -- (2.6, .4) -- (3.2, -.6) -- (4, -.6) node[above left] {$\vec{\ell}_2^{t}$};
    \draw[->] (2.4, -.6) -- (2.6, -.6) -- (3.2, .4) -- (4, .4) node[above left] {$\vec{\ell}_1^{t}$};

    \draw[semithick] (4, 0) rectangle (5.2, .8);
    \node at (4.6, .4) {$\mathcal{R}_1$};
    \draw[semithick] (4, -1) rectangle (5.2, -0.2);
    \node at (4.6, -.6) {$\mathcal{R}_2$};
    \draw[->] (5.2, .4) node[above right] {$\vec{x}^{t+1}$} -- (6.0, .4);
    \draw[->] (5.2, -.6) node[above right] {$\vec{y}^{t+1}$} -- (6.0, -.6);

    \node at (6.9, -0.1) {$\cdots$};
    \node at (-1.6, -0.1) {$\cdots$};
  \end{tikzpicture}
  \vspace{-3mm}
  \caption{Self-play method for computing NE in EFGs.}
  \label{fig:selfplay}
\end{figure}
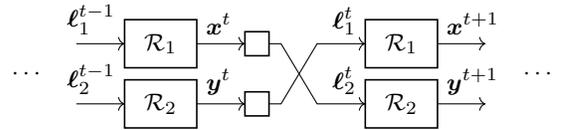

At each time $t$, the strategies $\vec{x}^t$ and $\vec{y}^t$ output by the regret minimizers are used to compute the loss vectors
\begin{equation}\label{eq:selfplay losses}
    \vec{\ell}_1^t \defeq \mat{A}_2\,\vec{y}^t,\quad \vec{\ell}_2^t \defeq -\mat{A}_2^\top \,\vec{x}^t.
\end{equation}
Let $\bar{\vec{x}}$ and $\bar{\vec{y}}$ be the average of the strategies output by $\mathcal{R}_1$ and $\mathcal{R}_2$, respectively, up to time $T$. Furthermore, let $R^T_1 \defeq \max_{\hat{\vec{x}}\in\cX} R_1^T(\hat{\vec{x}})$ and $R^T_2 \defeq \max_{\hat{\vec{y}}\in\cY} R_2^T(\hat{\vec{y}})$ be the maximum regret cumulated by $\mathcal{R}_1$ and $\mathcal{R}_2$ against any sequence-form strategy in $\cX$ and $\cY$, respectively.
A well-known folk lemma asserts that \[\xi(\bar{\vec{x}}, \bar{\vec{y}}) \le (R^T_1+ R^T_2)/T.\] So, if $\mathcal{R}_1$ and $\mathcal{R}_2$ have regret that grows sublinearly, then the strategy profile $(\bar{\vec{x}}, \bar{\vec{y}})$ converges to a saddle point.

%% file: text/stochastic.tex
\section{Stochastic Regret Minimization for Extensive-Form Games}

In this section we provide some key analytical tools to understand the performance of regret minimization algorithms when gradient estimates are used instead of exact gradient vectors. The results in this sections are complemented by those of~\cref{sec:estimators}, where we introduce computationally cheap gradient estimators for the purposes of equilibrium finding in extensive-form games.

\subsection{Regret Guarantees when Gradient Estimators are Used}\label{sec:martingale}
We start by studying how much the guarantee on the regret degrades when gradient estimators are used instead of exact gradient vectors. Our analysis need not assume that we operate over extensive-form strategy spaces, so we present our results in full generality.

Let $\tilde{\mathcal{R}}$ be a deterministic regret minimizer over a convex and compact set $\cZ$, and consider a second regret minimizer $\mathcal{R}$ over the same set $\cZ$ that is implemented starting from $\tilde{\mathcal{R}}$ as in \cref{fig:R tilde}. In particular, at all times $t$,
\begin{itemize}[leftmargin=7mm,nolistsep,itemsep=1mm]
    \item $\mathcal{R}$ queries the next decision $\vec{z}^t$ of $\tilde{\mathcal{R}}$, and outputs it;
    \item each gradient vector $\vec{\ell}^t$ received by $\mathcal{R}$ is used by $\mathcal{R}$ to compute a \emph{gradient estimate} $\tilde{\vec{\ell}}^t$ such that
        \[
            \bbE_t[\tilde{\vec{\ell}}^t] \defeq \bbE[\tilde{\vec{\ell}}^t \mid \tilde{\vec{\ell}}^1, \dots, \tilde{\vec{\ell}}^{t-1}] = \vec{\ell}^t.
        \]
        (that is, the estimate in unbiased).
        The internal regret minimizer $\tilde{\mathcal{R}}$ is then shown $\tilde{\vec{\ell}}^t$ instead of $\vec{\ell}^t$.
\end{itemize}

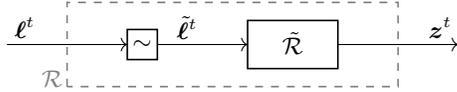
\begin{figure}[H]
    \centering
    \begin{tikzpicture}[scale=.8]
        \draw[semithick,gray,dashed] (-3, -.20) rectangle (2.5, 1.2);
        \draw[semithick] (0, .1) rectangle (1.5, .9);
        \node at (.75, .52) {$\tilde{\mathcal{R}}$};

        \draw[semithick] (-2, .25) rectangle (-1.5, .75);
        \node at (-1.75, .5) {$\sim$};

        \draw[->] (-4.0, .5) -- (-2, .5);
        \draw[->] (-1.5, .5) -- (0, .5);
        \draw[->] ( 1.5, .5) -- (3.5, .5);

        \node at (-3.7, .78) {\small$\vec{\ell}^t$};
        \node at (-1.0, .78) {\small$\tilde{\vec{\ell}}^t$};
        \node at (3.20, .78) {\small$\vec{z}^t$};
        \node[gray] at (-3.25, -0.07) {\small$\mathcal{R}$};
    \end{tikzpicture}\vspace{-4mm}
    \caption{Abstract regret minimizer considered in \cref{sec:martingale}.}
    \label{fig:R tilde}
\end{figure}

When the estimate of the gradient is very accurate (for instance, it have low variance), it is reasonable to expect that the regret $R^T$ incurred by $\mathcal{R}$ up to any time $T$ is roughly equal to the regret $\tilde{R}^T$ that is incurred by $\tilde{\mathcal{R}}$, plus some degradation term that depends on the variance of the estimates. We can quantify this relationship by fixing an arbitrary $\vec{u} \in \cZ$ and introducing the discrete-time stochastic process
\[
    d^t \defeq (\vec{\ell}^t)^{\!\top} (\vec{z}^t - \vec{u}) - (\tilde{\vec{\ell}}^t)^{\!\top} (\vec{z}^t - \vec{u}).
\]
Since by hypothesis $\bbE_t[\tilde{\vec{\ell}}^t] = \vec{\ell}^t$ and $\tilde{\mathcal{R}}$ is a deterministic regret minimizer, $\bbE_t[d^t]=0$ and so $\{d^t\}$ is a martingale difference sequence. This martingale difference sequence is well-known, especially in the context of \emph{bandit} regret minimization~\citep{Abernethy09:Beating,Bartlett08:High}. Using the classical Azuma-Hoeffding concentration inequality~\citep{Hoeffding63:Probability,Azuma67:Weighted}, we can prove the following.

\begin{restatable}{proposition}{propconcentration}\label{prop:concentration}
  Let $M$ and $\tilde{M}$ be positive constants such that $|(\vec{\ell}^t)^{\!\top}(\vec{z}-\vec{z}')| \le M$ and $|(\tilde{\vec{\ell}})^{\!\top}(\vec{z}-\vec{z}')| \le \tilde{M}$ for all times $t = 1, \dots, T$ and all feasible points $\vec{z}, \vec{z}' \in\cZ$. Then, for all $p \in (0, 1)$ and all $\vec{u} \in \cZ$,
\[
  \bbP\mleft[{R}^T(\vec{u}) \le \tilde{R}^T(\vec{u}) + (M + \tilde{M})\sqrt{2T \log\frac{1}{p}}\ \mright] \ge 1-p.
\]
\end{restatable}

Now if $\tilde{\mathcal{R}}$ has regret that grows sublinearly in $T$, a straightforward consequence of \cref{prop:concentration} is that $\mathcal{R}$ is also growing sublinearly in $T$ with high probability.

\subsection{Connection to Equilibrium Finding}\label{sec:martingale nash}

We now apply the general theory of \cref{sec:martingale} for the specific application of this paper---that is, Nash equilibrium computation in large extensive-form games.

We start from the construction of \cref{sec:blsp}. In particular, we instantiate two deterministic regret minimizers $\tilde{\mathcal{R}}_1$, $\tilde{\mathcal{R}}_2$ and let them play strategies against each other. However, instead of computing the exact losses $\vec{\ell}_1^t$ and $\vec{\ell}_2^t$ as in \eqref{eq:selfplay losses}, we compute their estimates $\tilde{\vec{\ell}}_1^t$ and $\tilde{\vec{\ell}}_2^t$ according to some algorithm that guarantees that $\bbE_t[\tilde{\vec{\ell}}_1^t] = {\vec{\ell}}_1^t$ and $\bbE_t[\tilde{\vec{\ell}}_2^t] = {\vec{\ell}}_2^t$ at all times $t$. We will show that despite this modification, the average strategy profile has a saddle point gap that is guaranteed to converge to zero with high probability.

Because of the particular definition of $\vec{\ell}_1^t$, we have that at all times $t$,
\begin{align*}
    \max_{\vec{x},\vec{x}'\in\cX}\Big|(\vec{\ell}_1^t)^{\!\top} (\vec{x} - \vec{x}')\Big| &=\! \max_{\vec{x},\vec{x}'\in\cX}\Big|(\vec{x}^t)^{\!\top}\! \mat{A}_2 \vec{y}^t - (\vec{x}')^{\!\top}\! \mat{A}_2 \vec{y}^t\Big| \\
   &= \Delta,
\end{align*}
where $\Delta$ is the payoff range of the game (see \cref{sec:efgs}).
(A symmetric statement holds for Player 2.)
For $i\in\{1,2\}$, let $\tilde{M}_i$ be positive constants such that $|(\tilde{\vec{\ell}}_i^t)^\top (\vec{z}-\vec{z}')| \le \tilde{M}_i$ at all times $t=1,\dots, T$ and all strategies $\vec{z},\vec{z}'$ in the sequence-form polytope for Player $i$ (that is, $\cX$ when $i=1$ and $\cY$ when $i=2$).
Using~\cref{prop:concentration}, we find that for all $\hat{\vec{x}} \in \cX$ and $\hat{\vec{y}} \in \cY$, with probability (at least) $1-p$,
\begin{align*}
    \sum_{t=1}^T (\vec{x}^t - \hat{\vec{x}})^\top \mat{A}_2\, \vec{y}^t &\le \tilde{R}_1^T(\hat{\vec{x}}) + (\Delta+\tilde{M}_1)\sqrt{2T \log\frac{1}{p}}\\
    -\sum_{t=1}^T (\vec{x}^t)^{\!\top}\!\! \mat{A}_2\, (\vec{y}^t - \hat{\vec{y}}) &\le \tilde{R}_2^T(\hat{\vec{y}}) + (\Delta+\tilde{M}_2)\sqrt{2T \log\frac{1}{p}}
\end{align*}
where $\tilde{R}_i$ denotes the regret of the regret minimizer $\tilde{\mathcal{R}}_i$ that at each time $t$ observes $\tilde{\vec{\ell}}^t_i$.

Summing the above inequality, dividing by $T$, and using the union bound we obtain that, with probability at least $1-2p$,
\begin{equation}\label{eq:forall x y}
    \begin{array}{l}
        \displaystyle\bar{\vec{x}}^\top \mat{A}_2\, \hat{\vec{y}} - \hat{\vec{x}}^\top \mat{A}_2\, \bar{\vec{y}} \le (\tilde{R}_1^T(\hat{\vec{x}}) + \tilde{R}_2^T(\hat{\vec{y}}))/T \\
        \displaystyle\hspace{2.8cm}+\ (2\Delta + \tilde{M}_1 + \tilde{M}_2)\sqrt{\frac{2}{T}\log\frac{1}{p}},
    \end{array}
\end{equation}
where $\bar{\vec{x}} \defeq \frac{1}{T}\sum_{t=1}^T \vec{x}^t$ and $\bar{\vec{y}} \defeq \frac{1}{T}\sum_{t=1}^T \vec{y}^t$. Since \eqref{eq:forall x y} holds for all $\hat{\vec{x}}\in \cX$ and $\hat{\vec{y}} \in \cY$, we obtain the following.

\begin{proposition}\label{prop:stochastic gap}
    With probability at least $1-2p$,
    \[
        \xi(\bar{\vec{x}}, \bar{\vec{y}}) \le \frac{\tilde{R}_1^T(\hat{\vec{x}}) + \tilde{R}_2^T(\hat{\vec{y}})}{T} + (2\Delta + \tilde{M}_1 + \tilde{M}_2)\sqrt{\frac{2}{T}\log\frac{1}{p}}.
    \]
\end{proposition}

If $\tilde{\mathcal{R}}_1$ and $\tilde{\mathcal{R}}_2$ have regret that is sublinear in $T$, then we conclude that the saddle point gap $\xi(\bar{\vec{x}}, \bar{\vec{y}})$ converges to $0$ with high probability like in the non-stochastic setting. So, $(\bar{\vec{x}},\bar{\vec{y}})$ converges to a saddle point over time.  

%% file: text/estimators.tex
\section{Game-Theoretic Gradient Estimators}\label{sec:estimators}

We complete the theory of \cref{sec:martingale,sec:martingale nash} by showing some examples of computationally cheap gradient estimators designed for game-theoretic applications. We will illustrate how each technique can be used to construct an estimate $\tilde{\vec{\ell}}_1^t$ for the gradient $\vec{\ell}_1^t = \mat{A}_2\,\vec{y}^t$ for Player~$1$ defined in \eqref{eq:selfplay losses}. The computation of an estimate for $\vec{\ell}_2^t$ is symmetrical.

\subsection{External Sampling}\label{sec:external sampling}
An unbiased estimator of the gradient vector $\vec{\ell}^t_1 = \mat{A}_2\, \vec{y}^t$ can be easily constructed by independently sampling \emph{pure} strategies $\tilde{\vec{y}}^t, \tilde{\vec{c}}^t$ for Player $2$ and the chance player, respectively. Indeed, as long as $\bbE_t[\tilde{\vec{y}}^t] = \vec{y}^t$ and $\bbE_t[\tilde{\vec{c}}^t] = \vec{c}$, from~\eqref{eq:expected payoff pl2} we have that for all $\vec{x} \in \cX$,
$
    \bar{u}_2(\vec{x}, \vec{y}^t, \vec{c}) = \bbE_t[\bar{u}_2(\vec{x}, \tilde{\vec{y}}^t, \tilde{\vec{c}}^t)].
$
Hence, the vector corresponding to the (random) linear function $\vec{x} \mapsto \bar{u}_2(\vec{x}, \tilde{\vec{y}}^t, \tilde{\vec{c}}^t)$ is an unbiased gradient estimator, called the \emph{external sampling} gradient estimator.

Since at all times $t$ $\tilde{\vec{y}}^t$ and $\tilde{\vec{c}}^t$ are sequence-form strategies, $\bar{u}_2(\vec{x}, \tilde{\vec{y}}^t, \tilde{\vec{c}}^t)$ is lower bounded by the minimum payoff of the game and upper bounded by the maximum payoff of the game. Hence, for this estimator $\tilde{M}$ in \cref{prop:concentration} is equal to the payoff range $\Delta$ of the game. Substituting that value into \cref{prop:stochastic gap}, we conclude that when the external sampling gradient estimator is used to estimate the gradient for both players, with probability at least $1-2p$ the saddle point gap of the average strategy profile $(\bar{\vec{x}},\bar{\vec{y}})$ is
\begin{equation}\label{eq:bound external}
   \xi(\bar{\vec{x}}, \bar{\vec{y}}) \le \frac{\tilde{R}_1^T(\hat{\vec{x}}) + \tilde{R}_2^T(\hat{\vec{y}})}{T} + 4\Delta\sqrt{\frac{2}{T}\log\frac{1}{p}}.
\end{equation}

The external sampling gradient estimator, that is, the vector corresponding to the linear function $\vec{x} \mapsto \bar{u}_2(\vec{x}, \tilde{\vec{y}}^t, \tilde{\vec{c}}^t)$, can be computed via a simple traversal of the game tree. The algorithm starts at the root of the game tree and starts visiting the tree. Every time a node that belongs to the chance player or to Player $2$ is encountered, an action is sampled according to the strategy $\vec{c}$ or $\vec{y}^t$, respectively. Every time a node for Player $1$ is encountered, the algorithm branches on all possible actions and recurses. A simple linear-time implementation is given in \cref{algo:external sampling}. For every node of Player $2$ or chance player, the algorithm branches on only one action, thus computing an external sampling gradient estimate is significantly cheaper to compute than the exact gradient $\vec{\ell}_1^t$.

\begin{algorithm}[ht]\small
  \caption{Efficient implementation of the external\hspace*{-1cm}\newline sampling gradient estimator}
    \label{algo:external sampling}\DontPrintSemicolon
    \KwIn{$\vec{y}^t$ strategy for Player $2$}
    \KwOut{$\tilde{\vec{\ell}}_1^t$ unbiased gradient estimate for $\vec{\ell}_1^t$ defined in~\eqref{eq:selfplay losses}}
    \BlankLine
    $\tilde{\vec{\ell}}_1^t \gets \vec{0} \in \bbR^{|\Sigma_1|}$\;
    \Subr{\normalfont\textsc{TraverseAndSample}($v$)}{
        $I \gets$ infoset to which $v$ belongs\;
        \uIf{\normalfont$v$ is a leaf}{
            $\tilde{\ell}^t_1[\sigma_1(v)] \gets u_1(v)$\;
        }\ElseIf{\normalfont$v$ belongs to the chance player}{
            Sample an action $a^* \sim \mleft(\frac{{c}[(I, a)]}{{c}[\sigma_{\chancepl}(I)]}\mright)_{a \in A_v}$\;
            $\textsc{TraverseAndSample}(\rho(v,a^*))$\;
        }
        \ElseIf{\normalfont$v$ belongs to Player $2$}{
            Sample an action $a^* \sim \mleft(\frac{{y^t}[(I, a)}{{y^t}[\sigma_{2}(I)]}
 \mright)_{a \in A_v}$\;
            $\textsc{TraverseAndSample}(\rho(v,a^*))$\;
        }
        \ElseIf{\normalfont$v$ belongs to Player $1$}{
            \For{$a \in A_v$}{
                $\textsc{TraverseAndSample}(\rho(v,a))$\;
            }
        }
    }
    \Hline{}
    $\textsc{TraverseAndSample}(r)$\Comment*{\color{commentcolor}$r$ is root of the game tree}
    \KwRet{$\tilde{\vec{\ell}}_1^t$}\;
\end{algorithm}

\textbf{Remark.} Analogous estimators where only the chance player's strategy $\vec{c}$ or only Player 2's strategy $\vec{y}^t$ are sampled are referred to as \emph{chance sampling} estimator and \emph{opponent sampling} estimator, respectively. In both cases, the same value of $\tilde{M} = \Delta$ (and therefore the bound in~\eqref{eq:bound external}) applies.

\textbf{Remark.} In the special case in which $\mathcal{R}_1$ and $\mathcal{R}_2$ run the CFR regret minimization algorithm, our analysis immediately implies the correctness of external-sampling MCCFR, chance-sampling MCCFR and opponent-sampling MCCFR, while at the same time yielding a significant improvement over the theoretical convergence to Nash equilibrium of the overall algorithm: the right hand side of~\eqref{eq:bound external} grows as $\sqrt{\log(1/p)}$ in $p$, compared to the $O(\sqrt{1/p})$ of the original analysis of by~\citet{Lanctot09:Monte}---an exponential improvement.

\subsection{Outcome Sampling}\label{sec:outcome sampling}
Let $\vec{w}^t \in \cX$ be an arbitrary strategy for Player $1$. Furthermore, let $\tilde{z}^t \in Z$ be a random variable such that for all $z\in Z$,
\[
    \bbP_t[\tilde{z}^t = z] = w^t[\sigma_1(z)]\cdot y^t[\sigma_2(z)]\cdot c[\sigma_\chancepl(z)],
\] and let $\vec{e}_z$ be defined as the vector such that $e_z[\sigma_1(z)] = 1$ and $e_z[\sigma] = 0$ for all other $\sigma \in \Sigma_1, \sigma \neq \sigma_1(z)$. It is a simple exercise to prove that the random vector
\[
  \tilde{\vec{\ell}}_1^t \defeq \frac{u_2(\tilde{z}^t)}{w^t[\sigma_1(\tilde{z}^t)]}\vec{e}_{\tilde{z}^t}
\]
is such that $\bbE_t[\tilde{\vec{\ell}}_1^t] = \vec{\ell}_1^t$ (see \cref{app:proofs} for a proof). This particular definition of $\tilde{\vec{\ell}}_1^t$ is called the \emph{outcome sampling} gradient estimator.

Computationally, the outcome sampling gradient estimator is cheaper than  the external sampling gradient estimator. Indeed, since $\vec{w}^t \in \cX$, one can sample $\tilde{z}^t$ by following a random path from the root of the game tree by sampling (from the appropriate player's strategy) one action at each node encountered along the way. The walk terminates as soon as it reaches a leaf, which corresponds to $\tilde{z}$.

As we show in \cref{app:proofs}, the value of $\tilde{M}$ for the outcome sampling gradient estimator is given by \[\tilde{M} = \Delta \cdot \max_{\sigma \in \Sigma_1} \frac{1}{w^t[\sigma]}.\] So, the high-probability bound on the saddle point gap is inversely proportional to the minimum entry in $w^t$, as already noted by~\citet{Lanctot09:Monte}.

In \cref{app:proofs} we show that a strategy $\vec{w}^*$ exists, such that $w^*[\sigma] \ge 1/(|\Sigma_1|-1)$ for every $\sigma\in\Sigma_1$.
Since $\vec{w}^*$ guarantees that all of the $|\Sigma_1|$ entries of $w^*$ are at least $1/(|\Sigma_1| - 1)$, we call $\vec{w}^*$ the \emph{balanced strategy}, and the corresponding outcome sampling regret estimator the \emph{balanced outcome sampling} regret estimator. As a consequence of the above analysis, when both players' gradients are estimated using the balanced outcome sampling regret estimator, with probability at least $1-2p$ the saddle point gap of the average strategy profile $(\bar{\vec{x}}, \bar{\vec{y}})$  is upper bounded as
\begin{equation*}
   \xi(\bar{\vec{x}}, \bar{\vec{y}}) \le \frac{\tilde{R}_1^T(\hat{\vec{x}}) + \tilde{R}_2^T(\hat{\vec{y}})}{T} + 2 (|\Sigma_1| + |\Sigma_2|)\Delta\sqrt{\frac{2}{T}\log\frac{1}{p}}.
\end{equation*}
To our knowledge, we are the first to introduce the balanced outcome sampling gradient estimator.

The final remark of \cref{sec:external sampling} applies to outcome sampling as well. 

%% file: text/experiments.tex
\section{Experiments}

In this section we perform numerical simulations to investigate the practical performance of several stochastic regret-minimization algorithms. First, we have the MCCFR algorithm instantiated with \emph{regret matching}~\citep{Hart00:Simple}.
Secondly, we instantiate two algorithms through our framework: FTRL and OMD, both using the dilated entropy DGF from \citet{Kroer20:Faster}, using their theoretically-correct recursive scheme for information-set weights\footnote{As opposed to constant information-set weights as used numerically by some past papers. Preliminary experiments with constant weights gave worse results.}.
We will show two sections of experiments, one with external sampling, and one with balanced outcome sampling.

\begin{figure*}
    \centering\begin{tabular}{cc}
        \includegraphics[scale=.7]{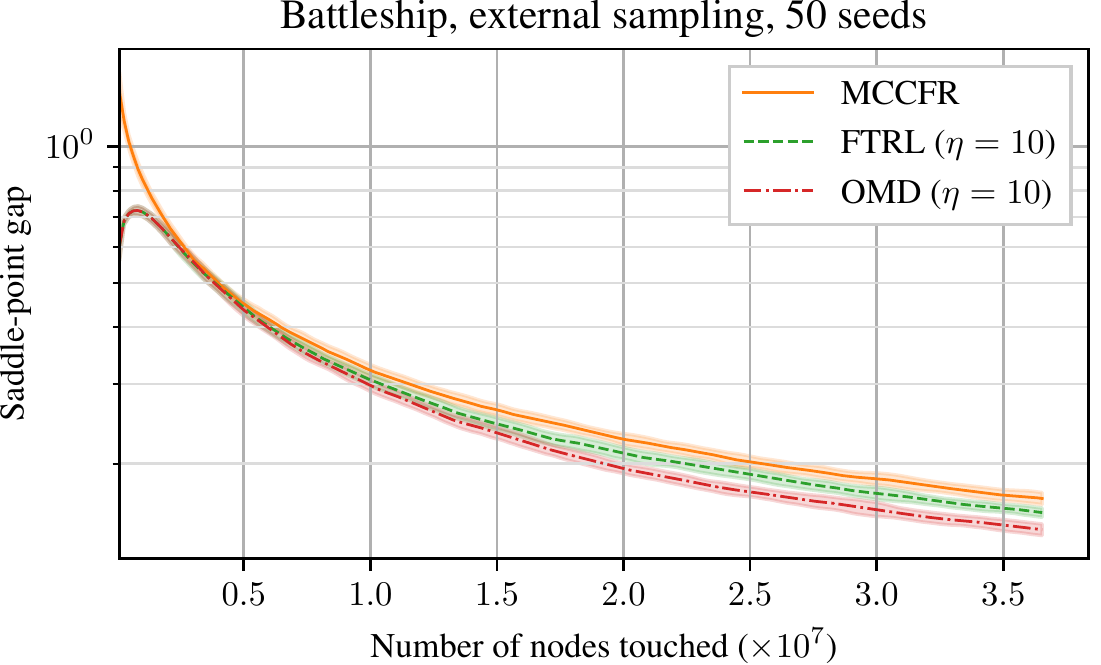} &
        \includegraphics[scale=.7]{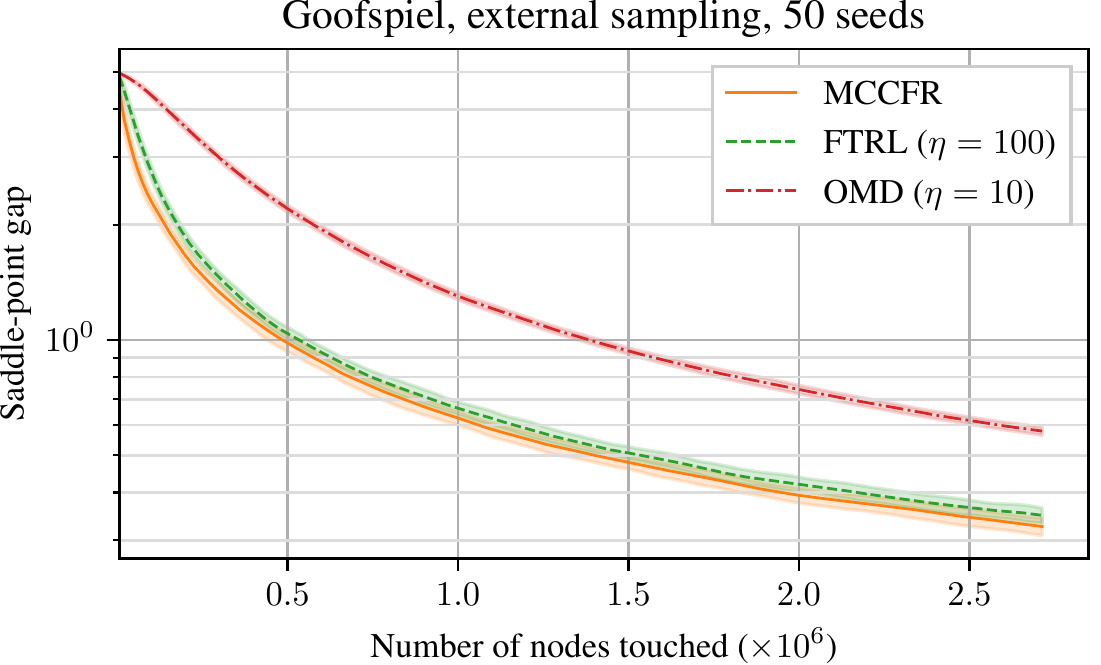}\\
        \includegraphics[scale=.7]{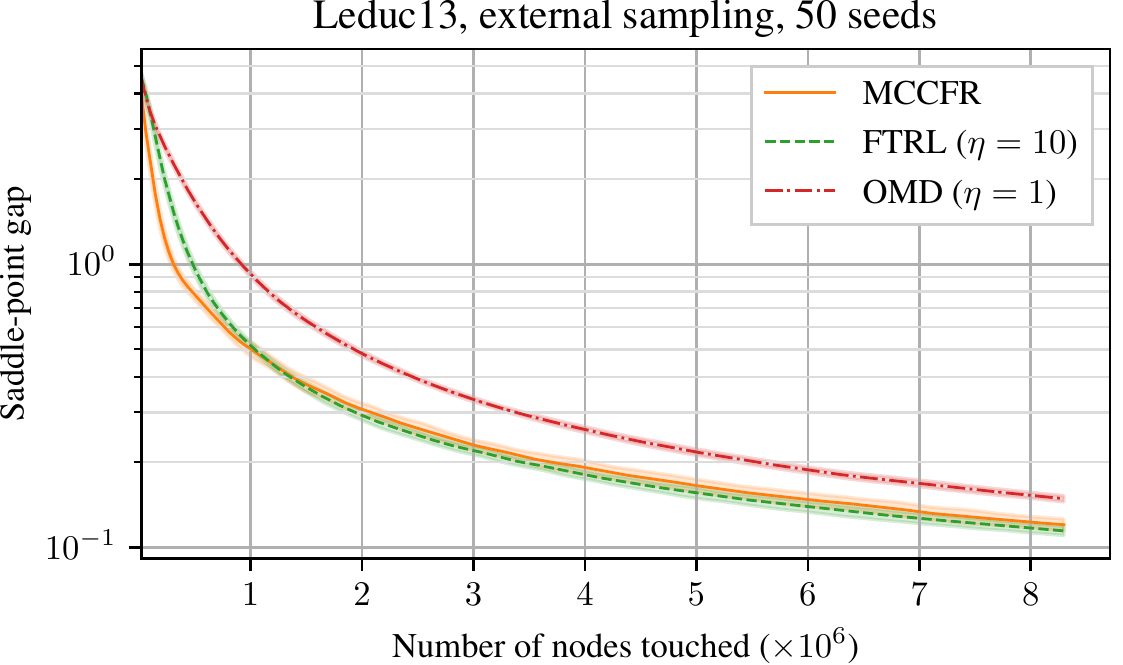} &
        \includegraphics[scale=.7]{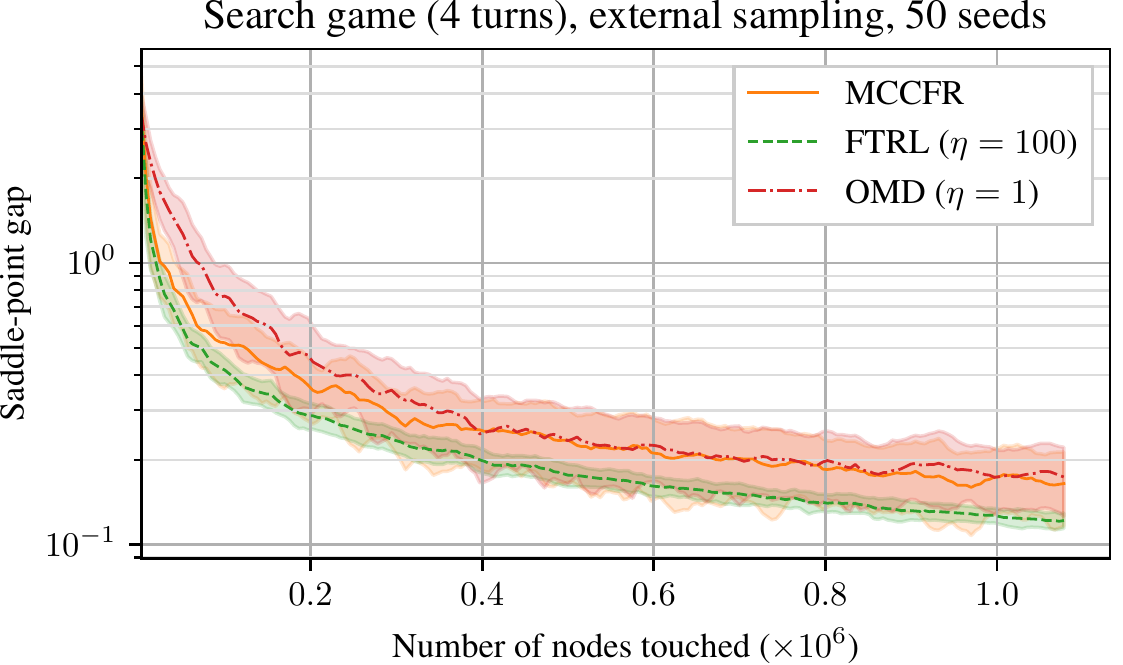}
    \end{tabular}
    \vspace{-2mm}
    \caption{Performance of MCCFR, FTRL, and OMD when using the external sampling gradient estimator.}
    \vspace{-4mm}
    \label{fig:experiments external}
\end{figure*}

For each game, we try four choices of stepsize $\eta$ in FTRL and OMD: $0.1, 1,
10, 100$. For each algorithm and game pair we show only the best-performing of
these four stepsizes in the plots below. The results for all stepsizes can be
found in \cref{app:additional experiments}. The stepsize is very
important: for most games where FTRL or OMD beats MCCFR, only the best stepsize
does so. At the same time, we did not extensively tune stepsizes (four stepsizes
increasing by a factor of 10 per choice leads to very coarse tuning), so there
is room for better tuning of these. Figuring out how to intelligently choose, or
adapt, stepsizes is an important follow-up work to the present paper, and would
likely lead to faster algorithms.

For each pair of game and algorithm we run the experiment 50 times, in order to account for variance in gradient estimates. All plots show the mean performance, and each line is surrounded by shading indicating one standard deviation around the mean performance.

In each plot we show the number of nodes (of the game tree) touched on the x-axis, and on the y-axis we show the saddle-point gap. All algorithms are run until the number of nodes touched corresponds to 50 full tree traversals (or equivalently 25 iterations of deterministic CFR or CFR$^+$).

We run our experiments on four different games. Below, we summarize some key properties of the games. The full description of each game is in \cref{app:games}.

\emph{Leduc poker} is a standard parametric benchmark game in the EFG-solving
community~\cite{Southey05:Bayes}. For the experiments we consider the largest variant of the game, denoted Leduc13. Leduc13 uses a deck of 13 unique
cards, with two copies of each card. The game has 166,336 nodes and 6,007 sequences per player.

\emph{Goofspiel} The variant of Goofspiel~\citep{Ross71:Goofspiel} that we use
in our experiments is a two-player card game, employing three identical decks
of 4 cards each.
This game has 54,421 nodes and 21,329 sequences per player.

\emph{Search} is a security-inspired pursuit-evasion game. The game is played on
a graph shown in Figure~\ref{fig:search_game} in \cref{app:games}. We consider two variants of the game, which differ on the number $k$ of simultaneous moves allowed before the game ties out. Search-4 uses $k = 4$ and has 21,613 nodes, 2,029 defender sequences, and 52
attacker sequences. Search-5 uses $k=5$ and has 87,972 nodes, 11,830 defender sequences, and 69 attacker sequences.
Our search game is a zero-sum variant of the one used by \citet{Kroer18:Robust}. A similar search game  considered by \citet{Bosansky14:Exact} and \citet{Bosansky15:Sequence}.

\emph{Battleship} is a parametric version of a classic board game, where two
competing fleets take turns shooting at each other~\citep{Farina19:Correlation}.
The game has 732,607 nodes, 73,130 sequences for Player 1,
and 253,940 sequences for Player 2.

\subsection{External Sampling}
For every algorithm and game we show the exploitability averaged over 50 runs (due to the variance in which gradient is sampled). Each line is surrounded by shading covering one standard deviation from the mean.

\cref{fig:experiments external} (top left) shows the performance on Battleship with external sampling. We see that both FTRL and OMD performs better than MCCFR when using a stepsize of $\eta=10$.
In Goofspiel (top right plot) we find that OMD performs significantly worse than MCCFR and FTRL. MCCFR performs slightly better than FTRL also.
In Leduc 13 (bottom left) we find that OMD performs significantly worse than MCCFR and FTRL. FTRL performs slightly better than MCCFR.
Finally, in Search-4 (bottom right) we find that OMD and MCCFR perform comparably, while FTRL performs significantly better.
Due to space limitations, we show the experimental evaluation for Search-5 in \cref{app:additional experiments}. In Search-5 all algorithms perform comparably, with FTRL performing slightly better than OMD and MCCFR.

Summarizing across all four games for external sampling, we see that FTRL, either with $\eta=10$ or $\eta=100$, was better than MCCFR on four out of five games (and essentially tied on the last game), with significantly better performance in the Search games. OMD seems to be more sensitive to stepsize, although it performs significantly better on Battleship.

\vspace{-1mm}
\subsection{Outcome Sampling}

Next, we investigate the performance of balanced outcome sampling. For that gradient estimator we performed 100 outcome samples per gradient estimate, and use the empirical mean of those 100 samples as our estimate. The reason for this is that FTRL and OMD seem more sensitive to stepsize issues under outcome sampling, and so we investigate a slightly more stable version of outcome sampling, while still having the same asymptotically cheap cost. It can be shown easily that by averaging gradient estimators, the constant $\tilde{M}$ required in \cref{prop:concentration} does not increase.

Due to computational time issues, we present performance for only 10 random seeds per game in outcome sampling. For this reason we omit performance on Search-4, which seemed too noisy to make conclusions about. Search-4 plots can be found in \cref{app:additional experiments}.

\begin{figure*}
    \centering\begin{tabular}{cc}
        \includegraphics[scale=.7]{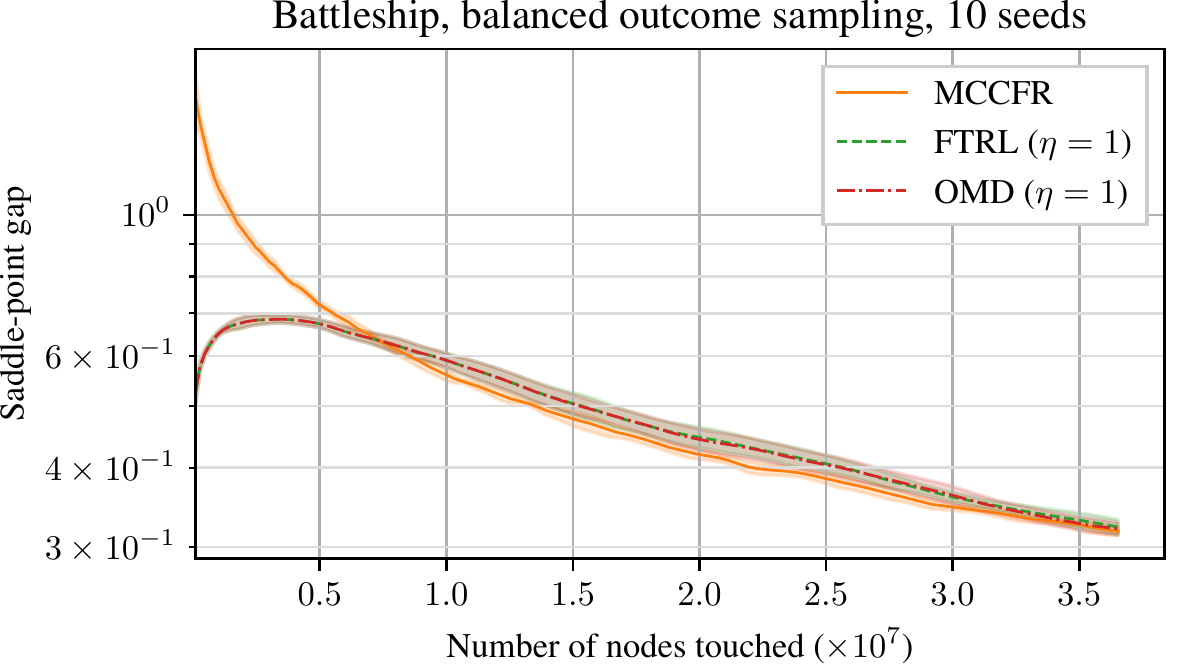} &
        \includegraphics[scale=.7]{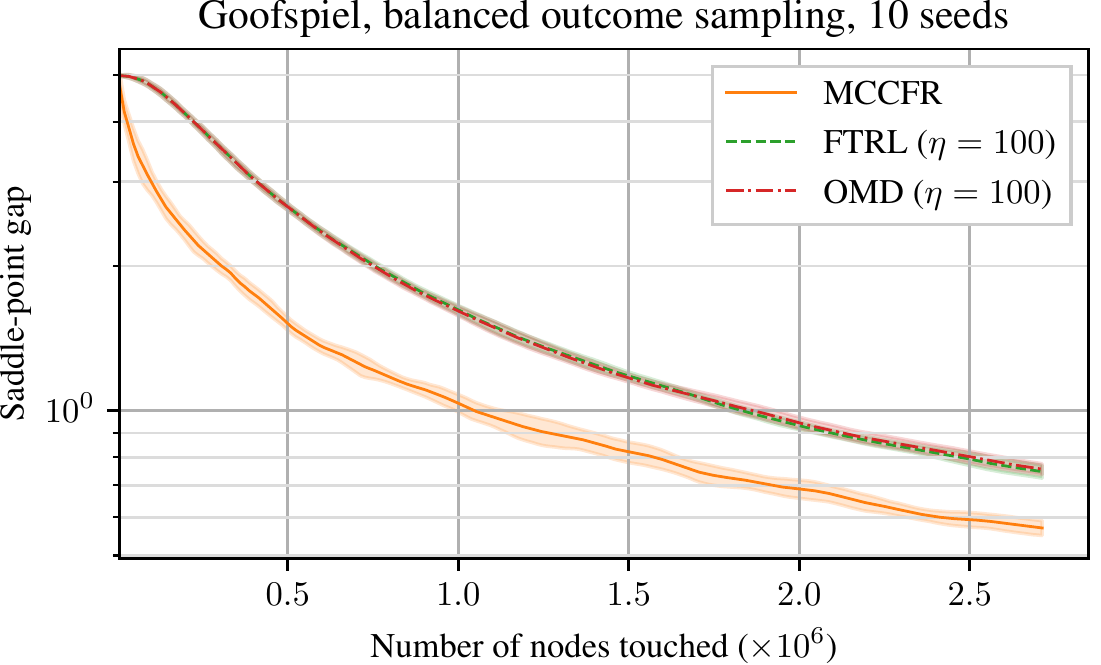}\\
        \includegraphics[scale=.7]{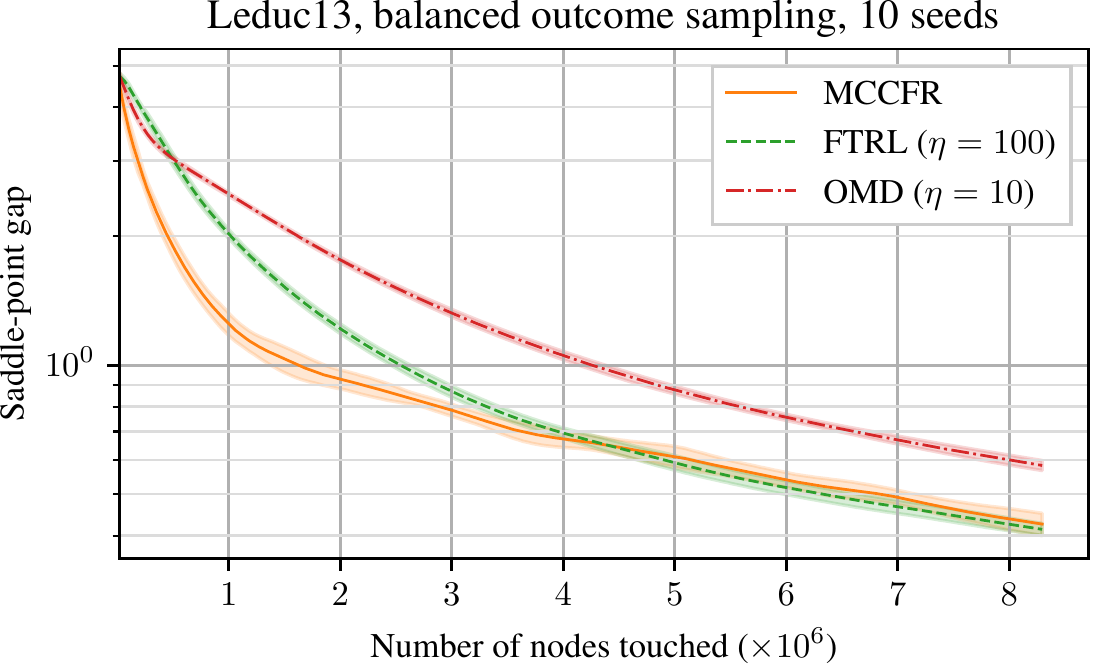} &
        \includegraphics[scale=.7]{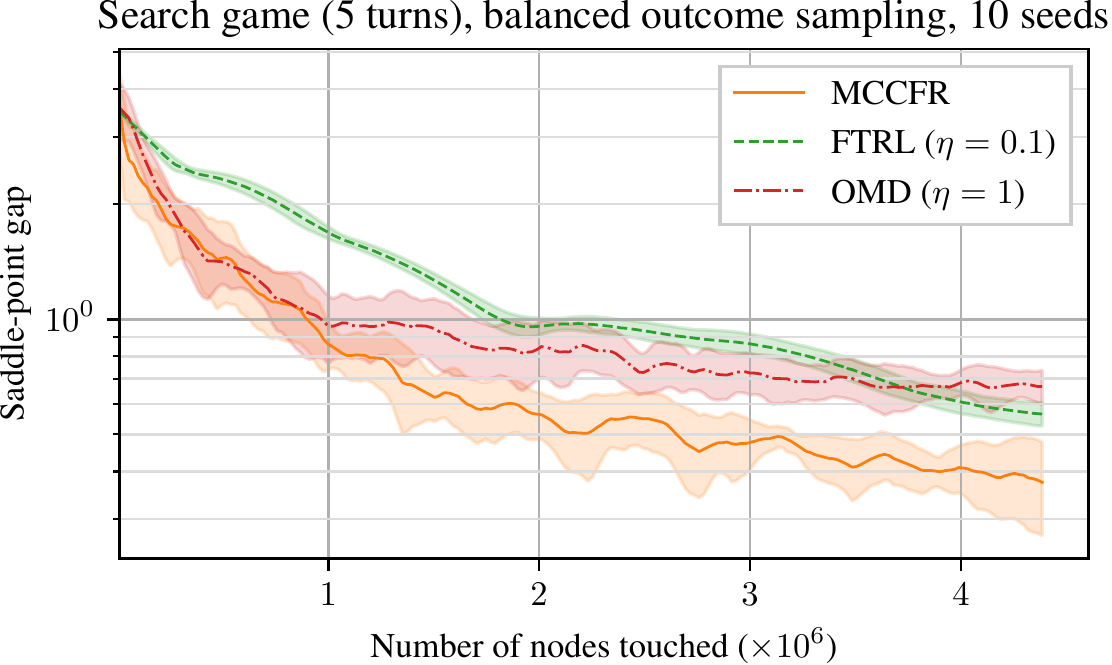}
    \end{tabular}
    \vspace{-2mm}
    \caption{Performance of MCCFR, FTRL, and OMD when using the balanced outcome sampling gradient estimator.}
    \label{fig:experiments outcome}
    \vspace{-1mm}
\end{figure*}

\cref{fig:experiments outcome} (top left) shows the performance on Battleship with outcome sampling. Here all algorithms perform essentially identically, with MCCFR performing better for a while, but then they all become similar around $3\times 10^7$ nodes touched.

In Goofspiel (top right) MCCFR performs significantly better than both FTRL and OMD. Both FTRL and OMD were best with $\eta=100$, our largest stepsize. It thus seems likely that even more aggressive stepsizes are needed in order to get better performance in Goofspiel.

In Leduc13 (bottom left) FTRL with outcome sampling is initially slower than MCCFR, but eventually overtakes it. OMD is significantly worse than the other algorithms.

Finally, in Search-5 (bottom right) MCCFR performs significantly better than FTRL and OMD,
although FTRL seems to be catching up in later iterations.

Overall, when the outcome sampling gradient estimator is used, MCCFR seems to perform better than FTRL and OMD. In two out of four games it is significantly better, in one it's marginally better, and in one FTRL is marginally better. We hypothesize that FTRL and OMD are much more sensitive to stepsize issues with outcome sampling as opposed to external sampling. This would make sense, as the variance becomes much higher.


%% file: text/conclusion.tex
\section{Conclusion}
\label{sec:conclusion}

We introduced a new framework for constructing stochastic regret-minimization methods for solving zero-sum games. This framework completely decouples the choice of regret minimizer and gradient estimator, thus allowing any regret minimizer to be coupled with any gradient estimator. Our framework yields a streamlined and dramatically simpler proof of MCCFR. Furthermore, it immediately gives a significantly stronger bound on the convergence rate of the MCCFR algorithm, whereby with probability $1-p$ the regret grows as $O(\sqrt{T\log(1/p)})$ instead of $O(\sqrt{T/p})$ as in the original analysis---an exponentially tighter bound. We also instantiated stochastic variants of the FTRL and OMD algorithms for solving zero-sum EFGs using our framework. Extensive numerical experiments showed that it is often possible to beat MCCFR using these algorithms, even with a mild amount of stepsize tuning. Due to its modular nature, our framework opens the door to many possible future research questions around stochastic methods for solving EFGs. Among the most promising are methods for controlling the stepsize in, for instance, FTRL or OMD, as well as instantiating our framework with other regret minimizers.

One potential avenue for future work is to develop gradient-estimation techniques with stronger control over the variance. In that case, it is possible to derive a variation of Proposition~\ref{prop:concentration} that is based on the sum of conditional variances, an \emph{intrinsic} notion of time in martingales (e.g., \citet{Blackwell73:Amount}). In particular, using the Freedman-style \citep{Freedman75:Tail} concentration result of
\citet{Bartlett08:High} for martingale difference sequences, we obtain:

\begin{restatable}{proposition}{propconcentrationlowvar}\label{prop:concentration low var}
  Let $T \ge 8$, and let $M$ and $\tilde{M}$ be positive constants such that $|(\vec{\ell}^t)^{\!\top}\vec{z}| \le M$ and $|(\tilde{\vec{\ell}})^{\!\top}\vec{z}| \le \tilde{M}$ for all times $t = 1, \dots, T$ and all feasible points $\vec{z}\in\cX$. Furthermore, let
$
    \sigma \defeq \sqrt{\var[d^t \mid \tilde{\vec{\ell}}^1, \dots, \tilde{\vec{\ell}}^{t-1}]}
$
be the square root of the sum of conditional variances.
Then, for all $p \in (0, 1)$ and all $\vec{u} \in \cX$,
\[
  \bbP\mleft[{R}^T(\vec{u}) \le \tilde{R}^T(\vec{u}) + 4\max\{\sigma\beta, (M + \tilde M)\beta^2\}\mright] \ge 1-p,
\]
where \[
    \beta \defeq \sqrt{\log\mleft(\frac{3\log T}{2p}\mright)}.\]
\end{restatable}
The concentration result of \cref{prop:concentration low var} takes into account the variance of the martingale difference sequences. When the variance is low, the dominant term in the right hand side of the inequality is $(M + \tilde{M})\beta^2 = O(\log\log T)$. On the other hand, when the variance is high (that is, $\sigma$ grows as $\sqrt{T}$), we recover a bound similar to the Azuma-Hoeffding inequality (albeit with a slightly worse polylog dependence on $T$).

Finally, our framework can also be applied to more general EFG-like problems, and thus this work also allows one to instantiate MCCFR or other stochastic methods for new sequential decision-making problems, for example by using the generalizations of CFR in~\citet{Farina19:Online} or~\citet{Farina19:Regret}.


%% file: text/appendix_proofs.tex
\section{Proofs}\label{app:proofs}

\subsection{Regret Guarantees when Gradient Estimators are
Used}

For completeness, we show a proof of \cref{prop:concentration}. As mentioned, it is an application of the Azuma-Hoeffding inequality for martingale difference sequences, which we now state (see, e.g., Theorem 3.14 of \citet{McDiarmid98:Concentration} for a proof).

\begin{theorem}[Azuma-Hoeffding inequality]
    \label{thm:azuma}
    Let $Y_1, \dots, Y_n$ be a martingale difference sequence with $a_k \le Y_k \le b_k$ for each $k$, for suitable constants $a_k, b_k$. Then for any $\tau \ge 0$,
\[
    \bbP\mleft[\sum Y_k \ge \tau \mright] \le e^{-2\tau^2 / \sum (b_k - a_k)^2}.
\]
\end{theorem}

\propconcentration*
\begin{proof}
    As observed in the body, $d^t \defeq (\vec{\ell}^t)^\top (\vec{z}^t - \vec{u}) - (\tilde{\vec{\ell}}^t)^\top (\vec{z}^t -\vec{u})$ is a martingale difference sequence. Furthermore, at all times $t$,
    \begin{align*}
        |d^t| &= |(\vec{\ell}^t)^\top (\vec{z}^t - \vec{u}) - (\tilde{\vec{\ell}}^t)^\top (\vec{z}^t -\vec{u})| \\
              &\le |(\vec{\ell}^t)^\top (\vec{z}^t - \vec{u})| + |(\tilde{\vec{\ell}}^t)^\top (\vec{z}^t -\vec{u})|\\
              &\le M + \tilde{M},
    \end{align*}
    and therefore $-(M + \tilde{M}) \le d^t \le (M + \tilde{M})$ for each $t$.

    Furthermore, 
    \[
        \sum_{t=1}^T d^t = \mleft(\sum_{t=1}^T (\vec{\ell}^t)^\top (\vec{z}^t - \vec{u})\mright) - \mleft(\sum_{t=1}^T (\tilde{\vec{\ell}}^t)^\top (\vec{z}^t - \vec{u})\mright) = R^T(\vec{u}) - \tilde{R}^T(\vec{u}).
    \]

    So, using \cref{thm:azuma}, for all $\tau \ge 0$
    \begin{align*}
        \bbP\mleft[ R^T(\vec{u}) \le \tilde{R}^T(\vec{u}) + \tau\mright]
        &= \bbP\mleft[ \sum_{t=1}^T d^t \le \tau \mright]\\
        &= 1-\bbP\mleft[ \sum_{t=1}^T d^t \ge \tau \mright]\\
        &\ge 1-\exp\mleft\{-\frac{2\tau^2}{\sum_{t=1}^T 4(M + \tilde{M})^2}\mright\}\\
        &= 1-\exp\mleft\{-\frac{2\tau^2}{4T (M + \tilde{M})^2}\mright\}.
    \end{align*}
    Finally, substituting $\tau = (M + \tilde{M})\sqrt{2T \log(1/p)}$ yields the statement.
\end{proof}

\subsection{Properties of the Outcome Sampling Gradient Estimator}

Let $\vec{w}^t \in \cX$ be an arbitrary strategy for Player $1$. Furthermore, let $\tilde{z}^t \in Z$ be a random variable such that for all $z\in Z$,
\[
    \bbP_t[\tilde{z}^t = z] = w^t[\sigma_1(z)]\cdot y^t[\sigma_2(z)]\cdot c[\sigma_\chancepl(z)],
\] and let $\vec{e}_z$ be defined as the vector such that $e_z[\sigma_1(z)] = 1$ and $e_z[\sigma] = 0$ for all other $\sigma \in \Sigma_1, \sigma \neq \sigma_1(z)$.

\begin{lemma}
The random vector
\[
  \tilde{\vec{\ell}}_1^t \defeq \frac{u_2(\tilde{z}^t)}{w^t[\sigma_1(\tilde{z}^t)]}\vec{e}_{\tilde{z}^t}
\]
is such that $\bbE_t[\tilde{\vec{\ell}}_1^t] = \vec{\ell}_1^t$.
\end{lemma}
\begin{proof} For all $\vec{x} \in \bbR^{|\Sigma_1|}$,
    \begin{align*}
        \bbE_t[\vec{\ell}_1^t]^\top \vec{x} &= \mleft(\sum_{z \in Z} \bbP[\tilde{z}^t = z] \cdot \frac{u_1(z)}{w^t[\sigma_1(z)]}\vec{e}_{z}\mright)^{\!\!\!\top}\! \vec{x}\\
        &= \mleft(\sum_{z \in Z} u_2(z)\cdot y^t[\sigma_2(z)]\cdot c[\sigma_\chancepl(z)]\cdot\vec{e}_{z}\mright)^{\!\!\!\top}\!\vec{x}\\
        &= \sum_{z \in Z} u_2(z)\cdot y^t[\sigma_2(z)]\cdot c[\sigma_\chancepl(z)]\cdot(\vec{e}_{z}^{\top}\vec{x})\\
        &= \sum_{z \in Z} u_2(z)\cdot y^t[\sigma_2(z)]\cdot c[\sigma_\chancepl(z)]\cdot x[\sigma_1(z)]\\
        &= u_2(\vec{x}, \vec{y}^t, \vec{c}) = \vec{\ell}_1^\top \vec{x}.
    \end{align*}
    Since the equality holds for all $\vec{x} \in \bbR^{|\Sigma_1|}$, we conclude $\bbE_t[\tilde{\vec{\ell}}_1^t] = \vec{\ell}_1$.
\end{proof}

Furthermore,

\begin{lemma}
    For all $\vec{x}, \vec{x}' \in \cX$,
    \[
        (\tilde{\vec{\ell}}_1)^\top (\vec{x} - \vec{x}') \le \Delta\cdot\max_{\sigma \in \Sigma_1} \frac{1}{w^t[\sigma]}.
    \]
\end{lemma}
\begin{proof}
    Using the definition of $\tilde{\vec{\ell}}_1$,
    \begin{align*}
        (\tilde{\vec{\ell}}_1)^\top (\vec{x} - \vec{x}') &= \frac{u_2(\tilde{z}^t)}{w^t[\sigma_1(\tilde{z}^t)]}\mleft( x[\sigma_1(\tilde{z}^t)] - x'[\sigma_1(\tilde{z}^t)] \mright).
    \end{align*}
    Since each entry of $\vec{x}$ and $\vec{x}'$ is in the interval $[0,1]$, the quantity $x[\sigma_1(\tilde{z}^t)] - x'[\sigma_1(\tilde{z}^t)]$ has absolute value in $[0,1]$ as well. Hence,
\[
\mleft|(\tilde{\vec{\ell}}_1)^\top (\vec{x} - \vec{x}')\mright| \le \max_{z \in Z} \mleft|\frac{u_2(z)}{w^t[\sigma_1(z)])}\mright| \le \Delta\cdot \max_{\sigma \in \Sigma_1} \frac{1}{w^t[\sigma]}
\]
as we wanted to show.
\end{proof}

\subsection{Balanced Strategy}

We now describe the construction of the \emph{balanced strategy} $\vec{w}^*$. Given $\sigma \in \Sigma_1$, we let $\mathcal{C}_\sigma$ be the set of information sets $I in \infos{1}$ such that $\sigma_1(I) = \sigma$. Furthermore, let $m_\sigma$, for $\sigma \in\Sigma_1$, be the number of terminal sequences in the subtree rooted under $\sigma$; formally, $m_\sigma$ is defined recursively as
\[
    m_\sigma = \begin{cases}
        1 & \text{if $\mathcal{C}_\sigma = \emptyseq$;}\\
        \displaystyle\sum_{I \in \mathcal{C}_\sigma}\sum_{a \in A_I} m_{(I, a)} & \text{otherwise.}
    \end{cases}
\]
Clearly, $m_\sigma \le |\Sigma_1| - 1$, since the empty sequence is never terminal (assuming Player $1$ acts at least once). With that, we define $\vec{w}^*$ such that $w^*[\emptyseq] = 1$ and that for all $\sigma=(I,a) \in \Sigma_1$,
\[
    w^*[\sigma] = \frac{m_\sigma}{\sum_{a' \in A_I}m_{(I, a')}} w^*[\sigma_1(I)].
\]
It is immediate to verify that $\vec{w}^*$ is indeed a valid sequence-form strategy. Furthermore, since for all $I \in \infos{1}$, $I \in \mathcal{C}_{\sigma_1(I)}$, we have
\[
    \sum_{a' \in A_I}m_{(I, a')} \le m_{\sigma(I)}.
\]
So,
\[
    w^*[\sigma] \ge \frac{m_\sigma}{m_{\sigma_1(I)}} w^*[\sigma_1(I)].
\]
By recursively expanding the definition of $w^*[\sigma_1(I)]$ on the right-hand side until $\sigma_1(I) = \emptyseq$, we ultimately obtain
\[
    w^*[\sigma] \ge \frac{1}{m_{\emptyseq}} \ge \frac{1}{|\Sigma_1| - 1}
\]
for all $\sigma$, as we wanted to show. 

%% file: text/appendix_games.tex
\section{Description of the Game Instances Used in the Experiments}\label{app:games}

We run our experiments on four different games, each described below.

\emph{Leduc poker} is a standard benchmark in the EFG-solving
community~\cite{Southey05:Bayes}. Our variant, Leduc 13, has a deck of 13 unique
cards, with two copies of each card. The game consists of two rounds. In the
first round, each player places an ante of $1$ in the pot and receives a single
private card. A round of betting then takes place with a two-bet maximum, with
Player 1 going first. A public shared card is then dealt face up and another
round of betting takes place. Again, Player 1 goes first, and there is a two-bet
maximum. If one of the players has a pair with the public card, that player
wins. Otherwise, the player with the higher card wins. All bets in the first
round are $1$, while all bets in the second round are $2$. This game has 166336
nodes and 6007 sequences per player.

\emph{Goofspiel} The variant of Goofspiel~\citep{Ross71:Goofspiel} that we use
in our experiments is a two-player card game, employing three identical decks
of 4 cards each. At the beginning of the game, each player receives one of the
decks to use it as its own hand, while the last deck is put face down between
the players, with cards in increasing order of rank from top to bottom. Cards
from this deck will be the prizes of the game. In each round, the players
privately select a card from their hand as a bet to win the topmost card in the
prize deck. The selected cards are simultaneously revealed, and the highest one
wins the prize card. In case of a tie, the prize card is discarded. Each prize
card’s value is equal to its face value, and at the end of the game the players’
score are computed as the sum of the values of the prize cards they have won.
This game has 54421 nodes and 21329 sequences per player.

\emph{Search} is a security-inspired pursuit-evasion game. The game is played on
the graph shown in Figure~\ref{fig:search_game}.

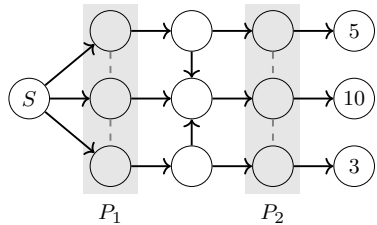
\begin{figure}[H]
  \centering
  \scalebox{0.9}{
  \input{figs/search_game}
  }
  \vspace{-3mm}
  \caption{The graph on which the search game is played.}
  \label{fig:search_game}
  \vspace{-1mm}
\end{figure}

It is a simultaneous-move game
(which can be modeled as a turn-taking EFG with appropriately chosen information
sets). The defender controls two patrols that can each move within their
respective shaded areas (labeled P1 and P2). At each time step the controller
chooses a move for both patrols. The attacker is always at a single node on the
graph, initially the leftmost node labeled $S$. The attacker can move freely to
any adjacent node (except at patrolled nodes, the attacker cannot move from a
patrolled node to another patrolled node). The attacker can also choose to wait
in place for a time step in order to clean up their traces. If a patrol visits a
node that was previously visited by the attacker, and the attacker did not wait
to clean up their traces, they can see that the attacker was there. If the
attacker reaches any of the rightmost nodes they receive the respective payoff
at the node ($5$, $10$, or $3$, respectively). If the attacker and any patrol
are on the same node at any time step, the attacker is captured, which leads to
a payoff of $-1$ for the attacker and a payoff of $1$ for the defender. Finally,
the game times out after $k$ simultaneous moves, in which case both players defender
receive payoffs $0$. Search-4
(Search-5) has 21613 (87,927) nodes, 2029 (11,830) defender sequences, and 52
(69) attacker sequences.

Our search game is a zero-sum variant of the one used by \citet{Kroer18:Robust}. A similar search game  considered by \citet{Bosansky14:Exact} and \citet{Bosansky15:Sequence}.

\emph{Battleship} is a parametric version of a classic board game, where two
competing fleets take turns shooting at each other~\citep{Farina19:Correlation}.
At the beginning of the game, the players take turns at secretly placing a set
of ships on separate grids (one for each player) of size $3\times 2$. Each ship
has size 2 (measured in terms of contiguous grid cells) and a value of 1, and
must be placed so that all the cells that make up the ship are fully contained
within each player’s grids and do not overlap with any other ship that the
player has already positioned on the grid. After all ships have been placed. the
players take turns at firing at their opponent. Ships that have been hit at all
their cells are considered sunk. The game continues until either one player has
sunk all of the opponent’s ships, or each player has completed r shots. At the
end of the game, each player’s payoff is calculated as the sum of the values of
the opponent’s ships that were sunk, minus the sum of the values of ships which
that player has lost. The game has 732607 nodes, 73130 sequences for player 1,
and 253940 sequences for player 2. 

%% file: figs/search_game.tex
\begin{tikzpicture}
\path[fill=black!10!white] (.8, -1.4) rectangle (1.6, 1.4);
\path[fill=black!10!white] (3.2, -1.4) rectangle (4.0, 1.4);
\node at (1.2, -1.7) {$P_1$};
\node at (3.6, -1.7) {$P_2$};

  \node[draw, circle, minimum width=.6cm, inner sep=0] (A) at (0, 0) {$S$};
  \node[draw, circle, minimum width=.6cm] (B) at (1.2, 1) {};
  \node[draw, circle, minimum width=.6cm] (C) at (1.2, 0) {};
  \node[draw, circle, minimum width=.6cm] (D) at (1.2, -1) {};
    \node[draw, circle, minimum width=.6cm] (E) at (2.4, 1) {};
    \node[draw, circle, minimum width=.6cm] (F) at (2.4, 0) {};
    \node[draw, circle, minimum width=.6cm] (G) at (2.4, -1) {};
      \node[draw, circle, minimum width=.6cm] (H) at (3.6, 1) {};
      \node[draw, circle, minimum width=.6cm] (I) at (3.6, 0) {};
      \node[draw, circle, minimum width=.6cm] (J) at (3.6, -1) {};
            \node[draw, circle, minimum width=.6cm,inner sep=0] (K) at (4.8, 1) {$5$};
            \node[draw, circle, minimum width=.6cm,inner sep=0] (L) at (4.8, 0) {$10$};
            \node[draw, circle, minimum width=.6cm,inner sep=0] (M) at (4.8, -1) {$3$};

\draw[thick,->] (A) edge (B);
\draw[thick,->] (A) edge (C);
\draw[thick,->] (A) edge (D);
\draw[thick,->] (B) edge (E);
\draw[thick,->] (C) edge (F);
\draw[thick,->] (D) edge (G);
\draw[thick,->] (E) edge (F);
\draw[thick,->] (G) edge (F);
\draw[thick,->] (E) edge (H);
\draw[thick,->] (F) edge (I);
\draw[thick,->] (G) edge (J);
\draw[thick,->] (H) edge (K);
\draw[thick,->] (I) edge (L);
\draw[thick,->] (J) edge (M);

\draw[thick,gray,dashed] (B) edge (C);
\draw[thick,gray,dashed] (C) edge (D);

\draw[thick,gray,dashed] (H) edge (I);
\draw[thick,gray,dashed] (I) edge (J);
\end{tikzpicture}

%% file: text/appendix_experiments.tex
\section{Additional Experimental Results}
\label{app:additional experiments}

\subsection{External Sampling}
The Search-5 plot omitted from the main paper is shown here.
\begin{figure}[H]
  \centering
  \includegraphics[width=0.49\columnwidth]{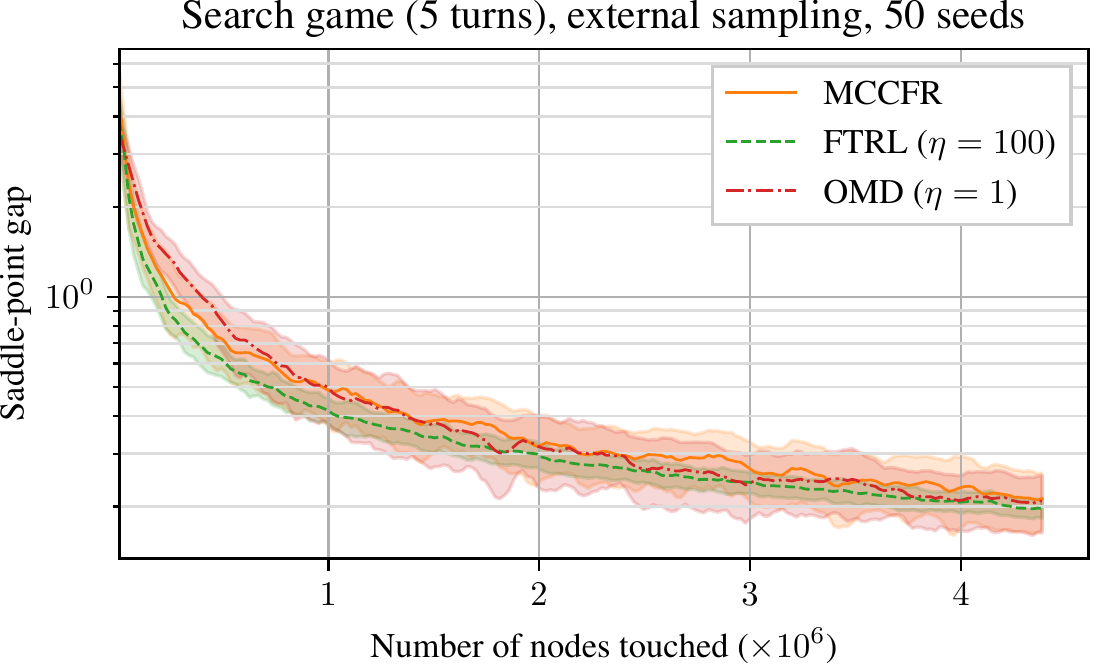}
  \vspace{-2mm}
  \caption{Performance of MCCFR, FTRL, and OMD with external sampling on Search-5.}
  \label{fig:search5_external}
  \vspace{-2mm}
\end{figure}

Figures~\ref{fig:appendix bs external} through~\ref{fig:appendix search5 external} show the performance of FTRL and OMD for all four stepsizes that we tried on each game: $\eta=0.1,1,10,100$.

\begin{figure}[H]
  \centering
  \includegraphics[width=0.49\columnwidth]{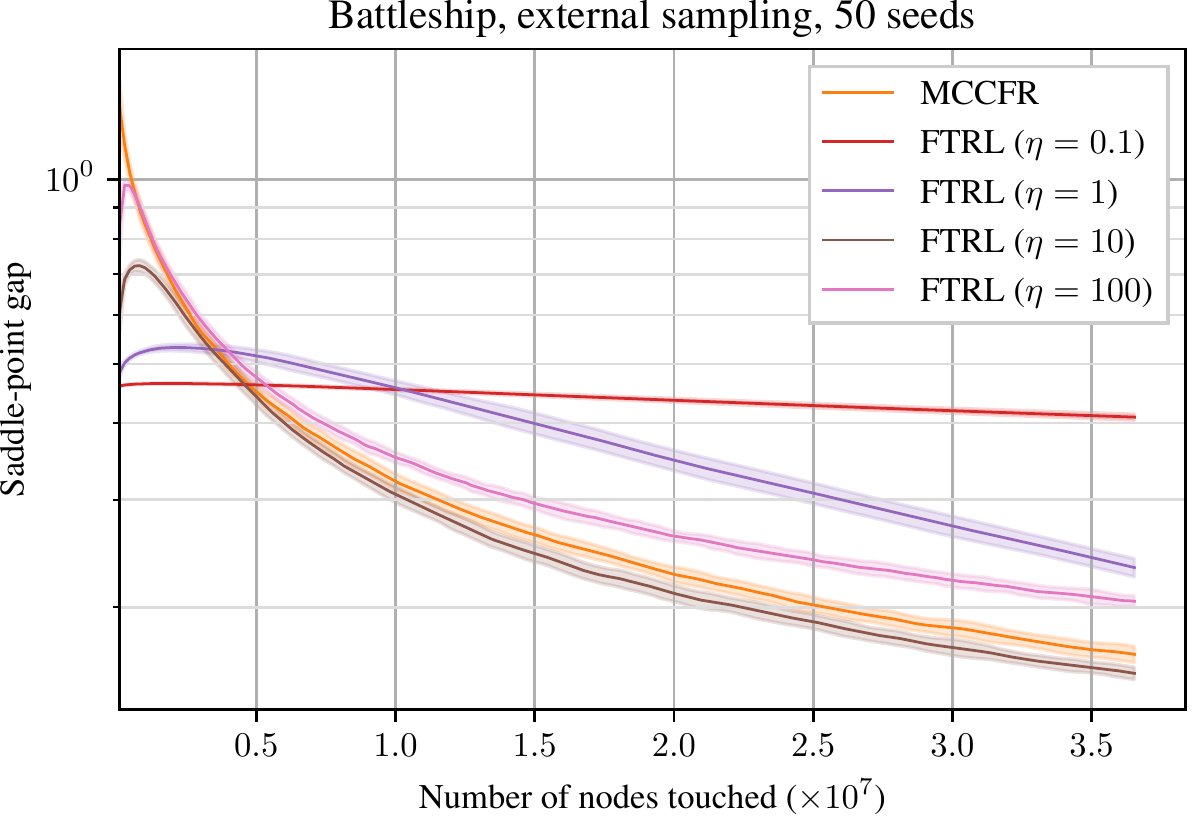}
  \includegraphics[width=0.49\columnwidth]{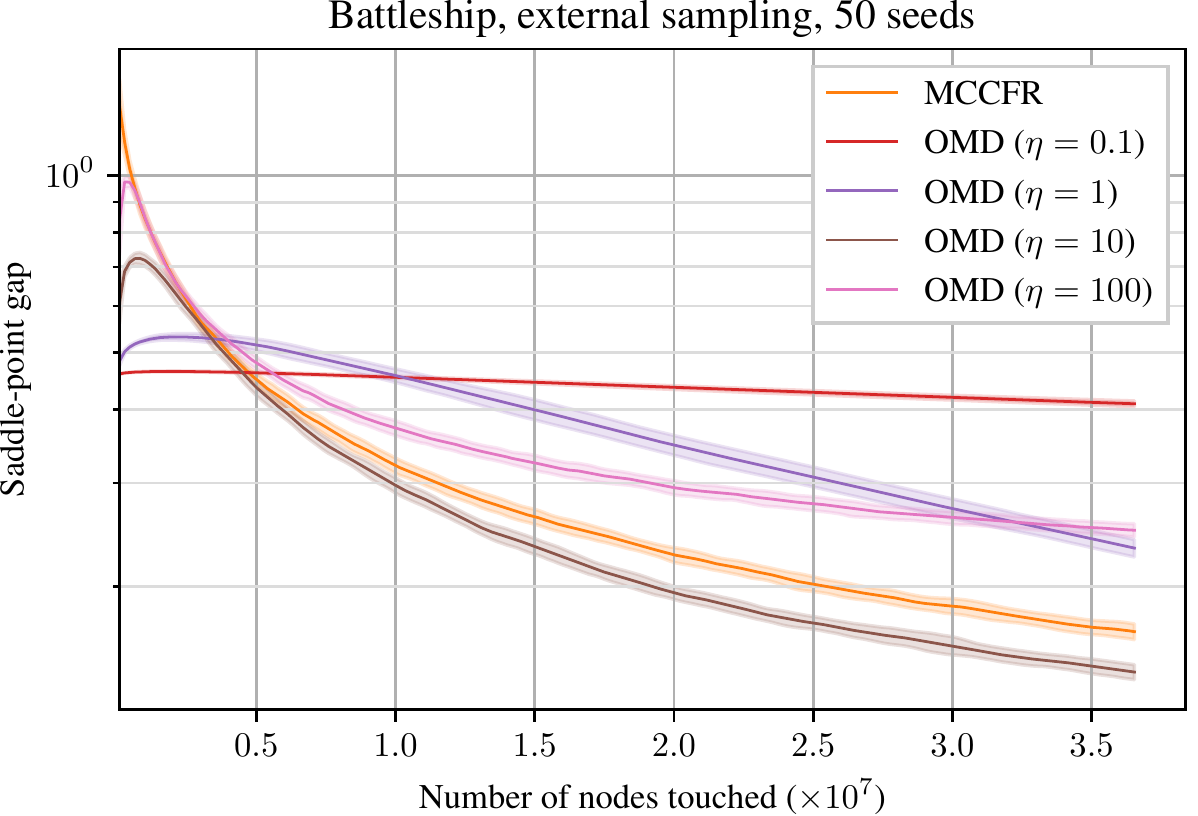}
  \vspace{-2mm}
  \caption{Performance of FTRL and OMD with four stepsizes on Battleship with external sampling. MCCFR shown for reference}
  \vspace{-2mm}
  \label{fig:appendix bs external}
\end{figure}
\begin{figure}[H]
  \centering
  \includegraphics[width=0.49\columnwidth]{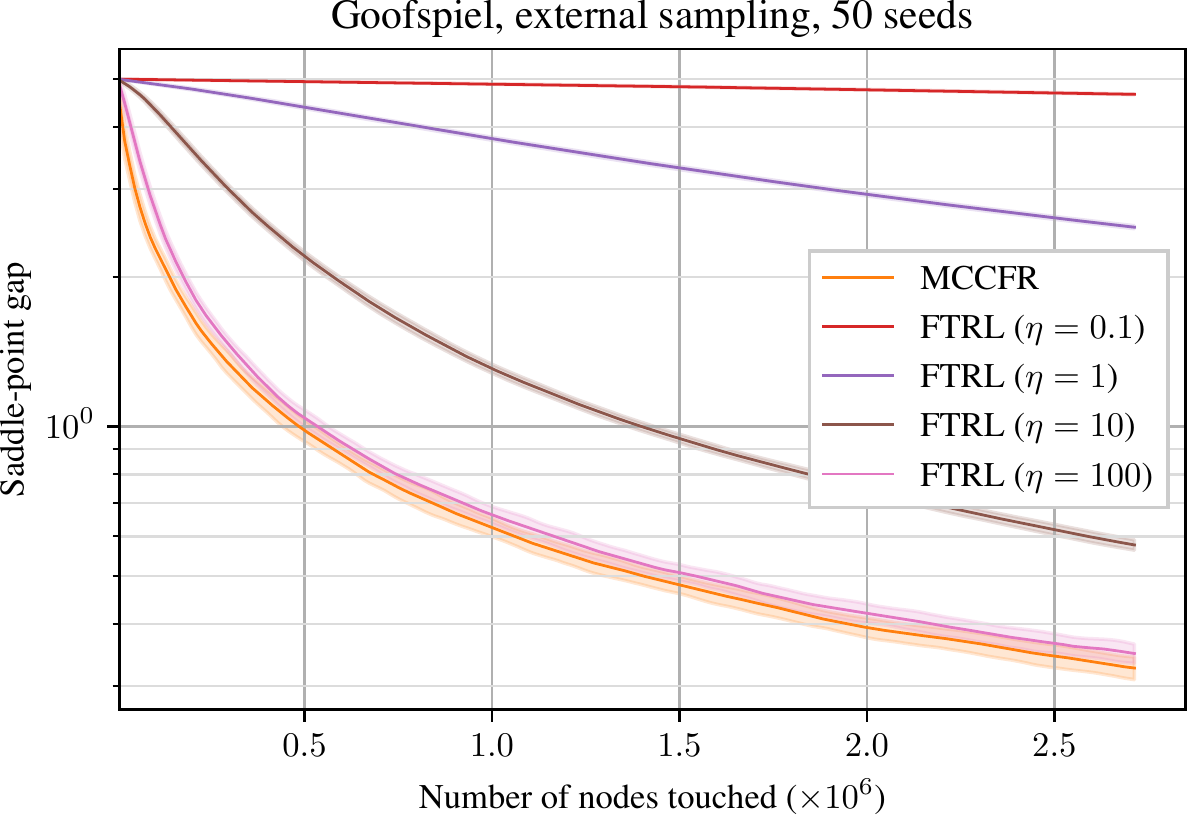}
  \includegraphics[width=0.49\columnwidth]{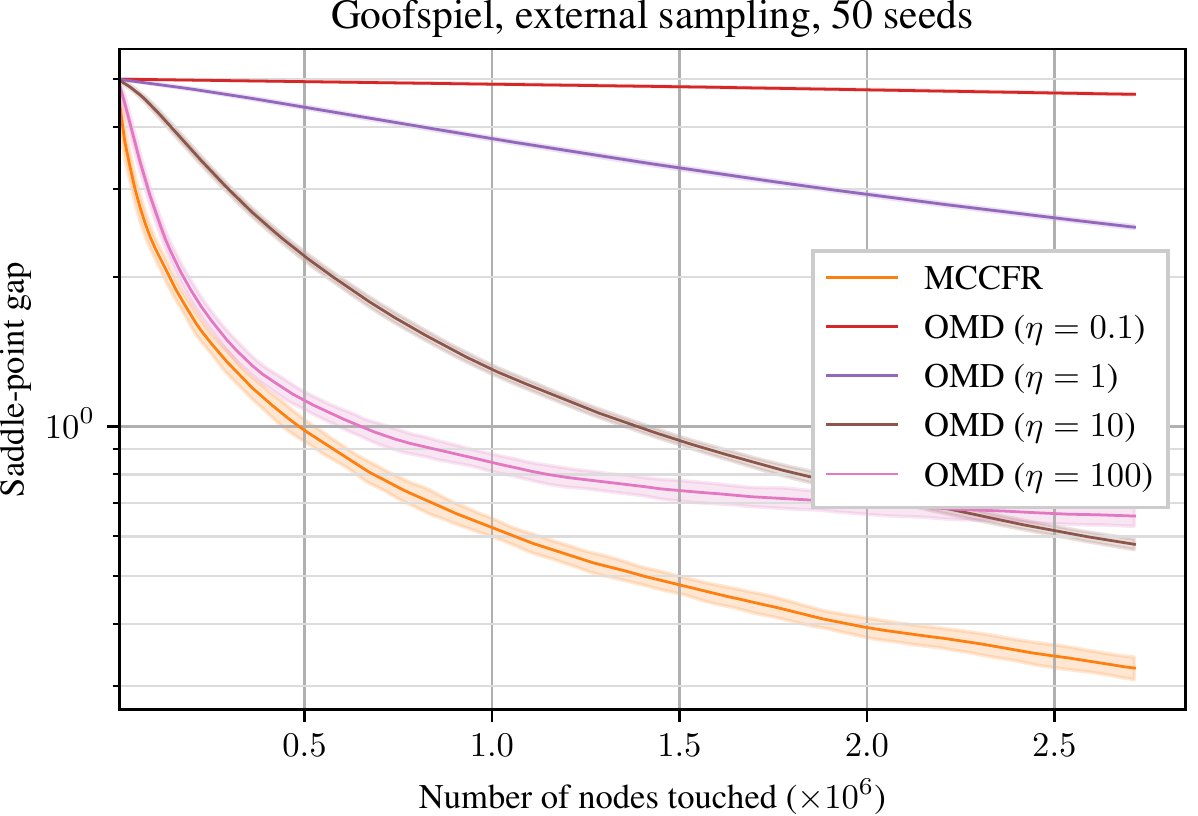}
  \vspace{-2mm}
  \caption{Performance of FTRL and OMD with four stepsizes on Goofspiel with external sampling. MCCFR shown for reference}
  \vspace{-2mm}
  \label{fig:appendix goofspiel external}
\end{figure}
\begin{figure}[H]
  \centering
  \includegraphics[width=0.49\columnwidth]{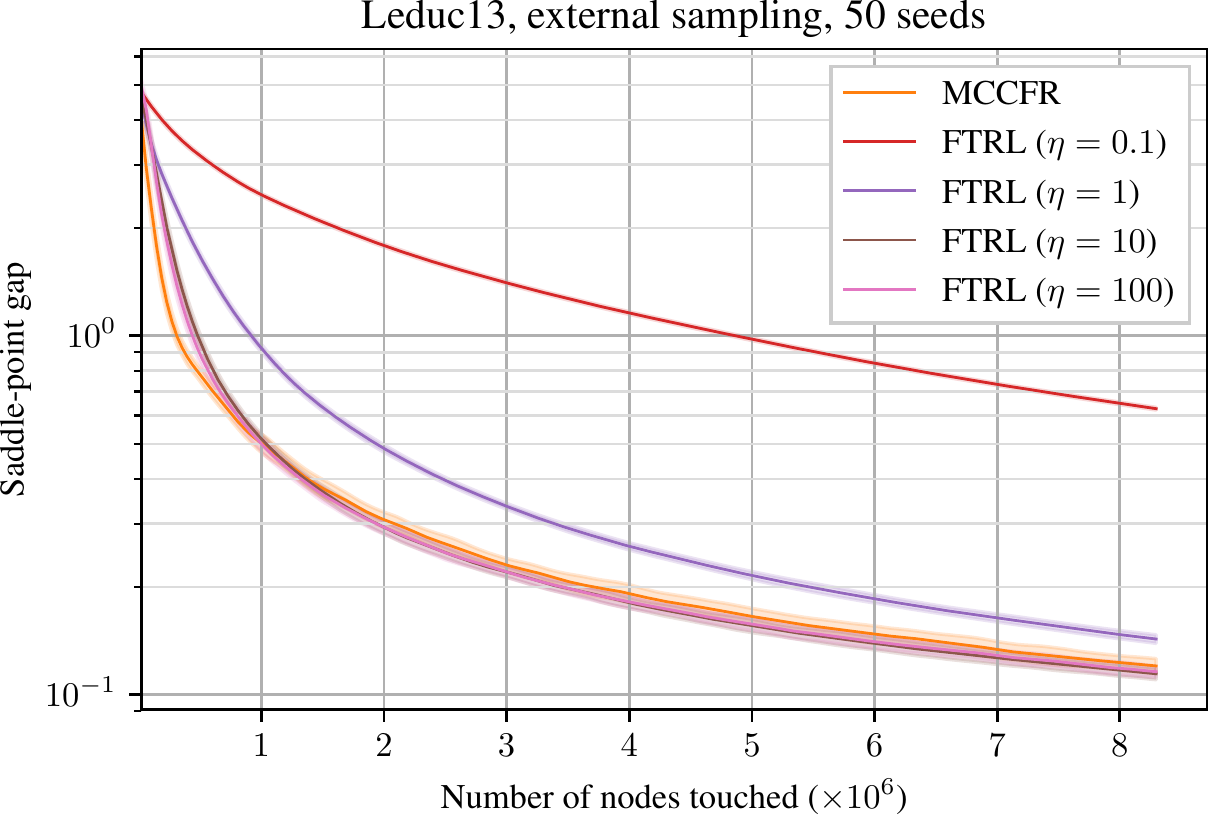}
  \includegraphics[width=0.49\columnwidth]{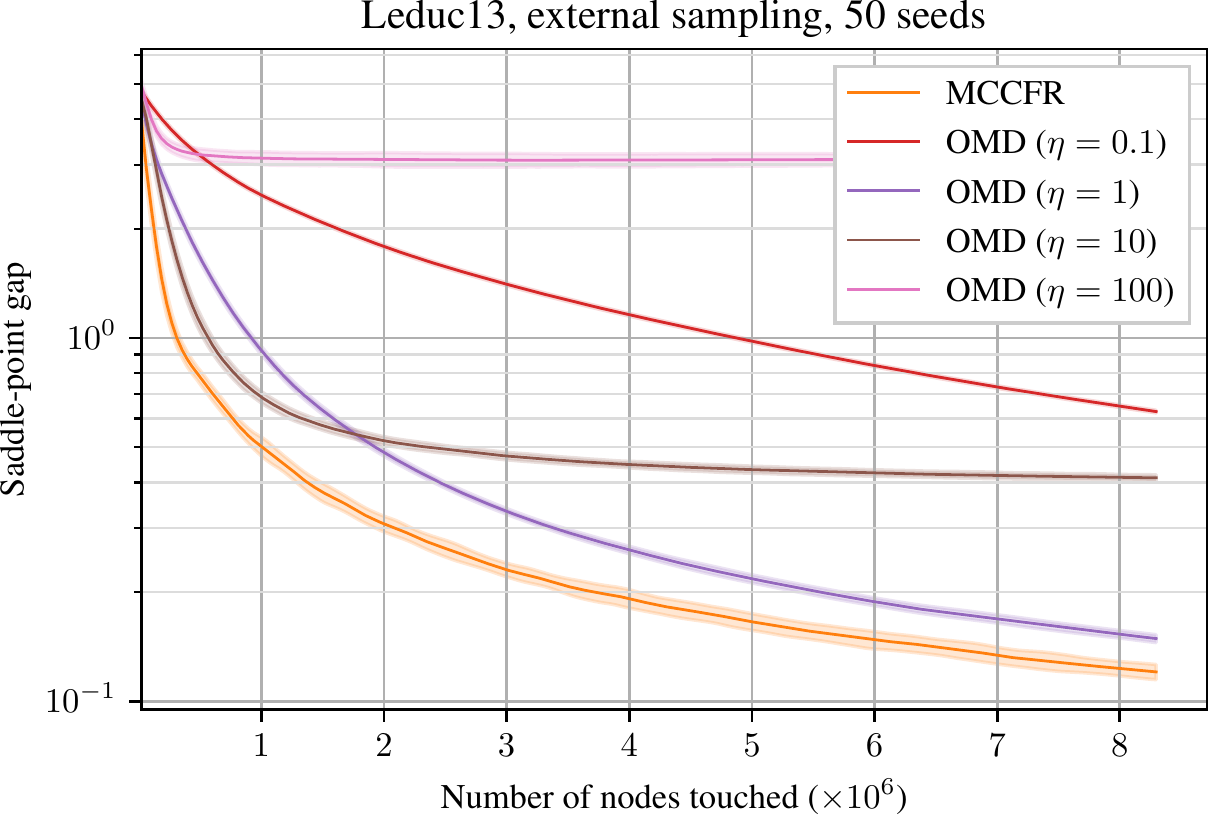}
  \caption{Performance of FTRL and OMD with four stepsizes on Leduc 13 with external sampling. MCCFR shown for reference}
  \label{fig:appendix leduc external}
\end{figure}
\begin{figure}[H]
  \centering
  \includegraphics[width=0.49\columnwidth]{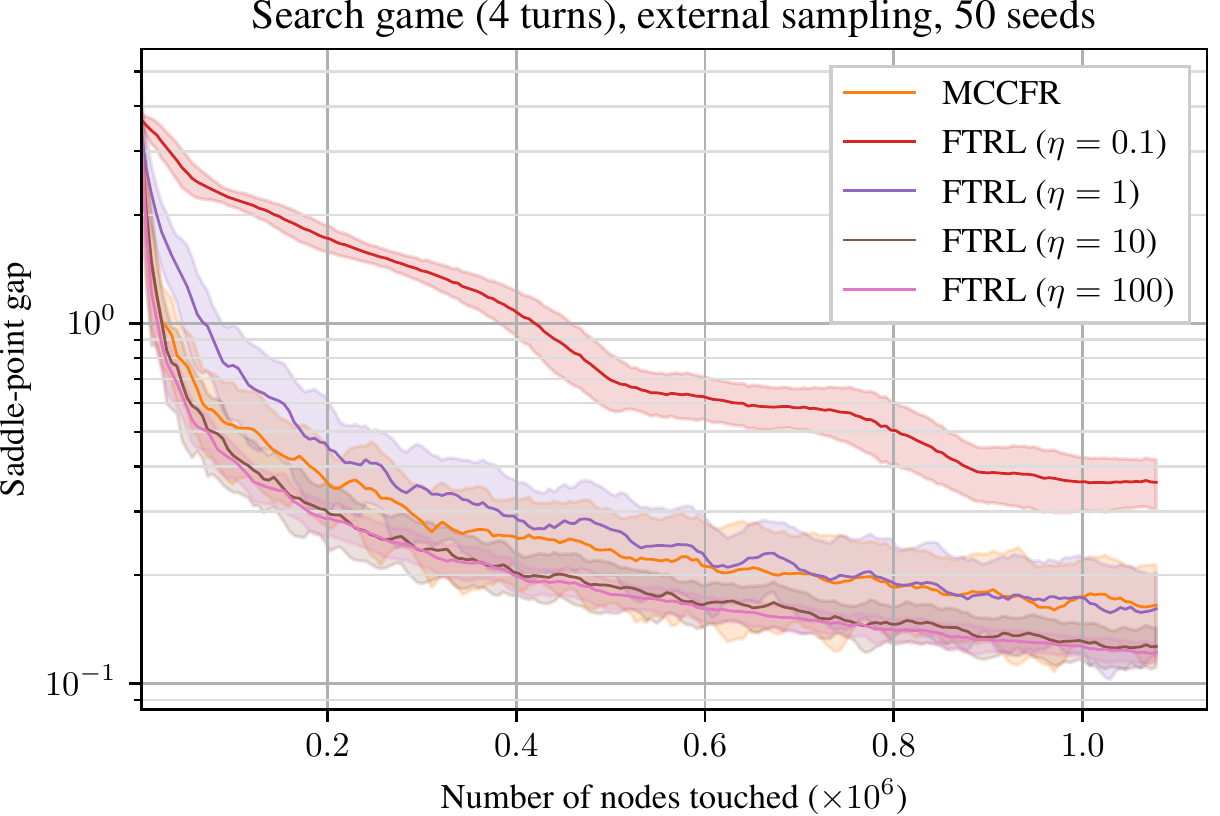}
  \includegraphics[width=0.49\columnwidth]{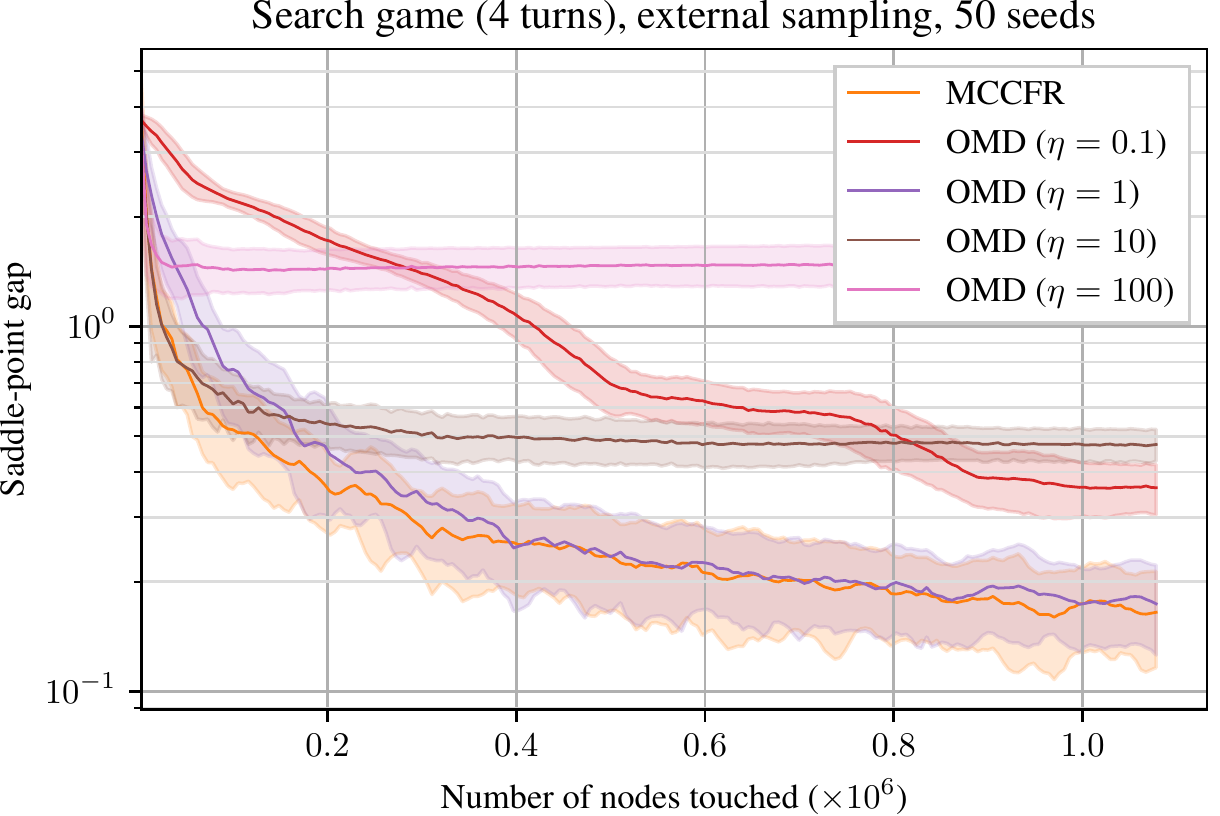}
  \caption{Performance of FTRL and OMD with four stepsizes on Search-4 with external sampling. MCCFR shown for reference}
  \label{fig:appendix search4 external}
\end{figure}
\begin{figure}[H]
  \centering
  \includegraphics[width=0.49\columnwidth]{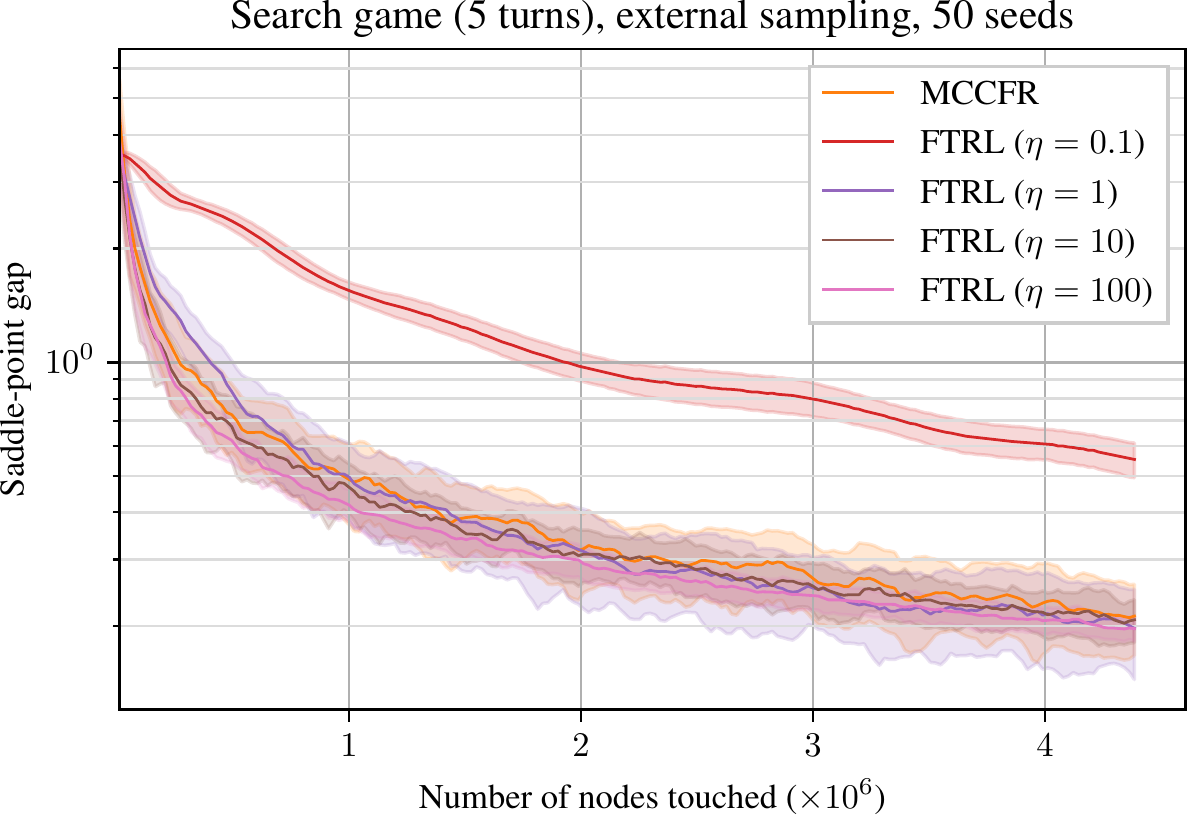}
  \includegraphics[width=0.49\columnwidth]{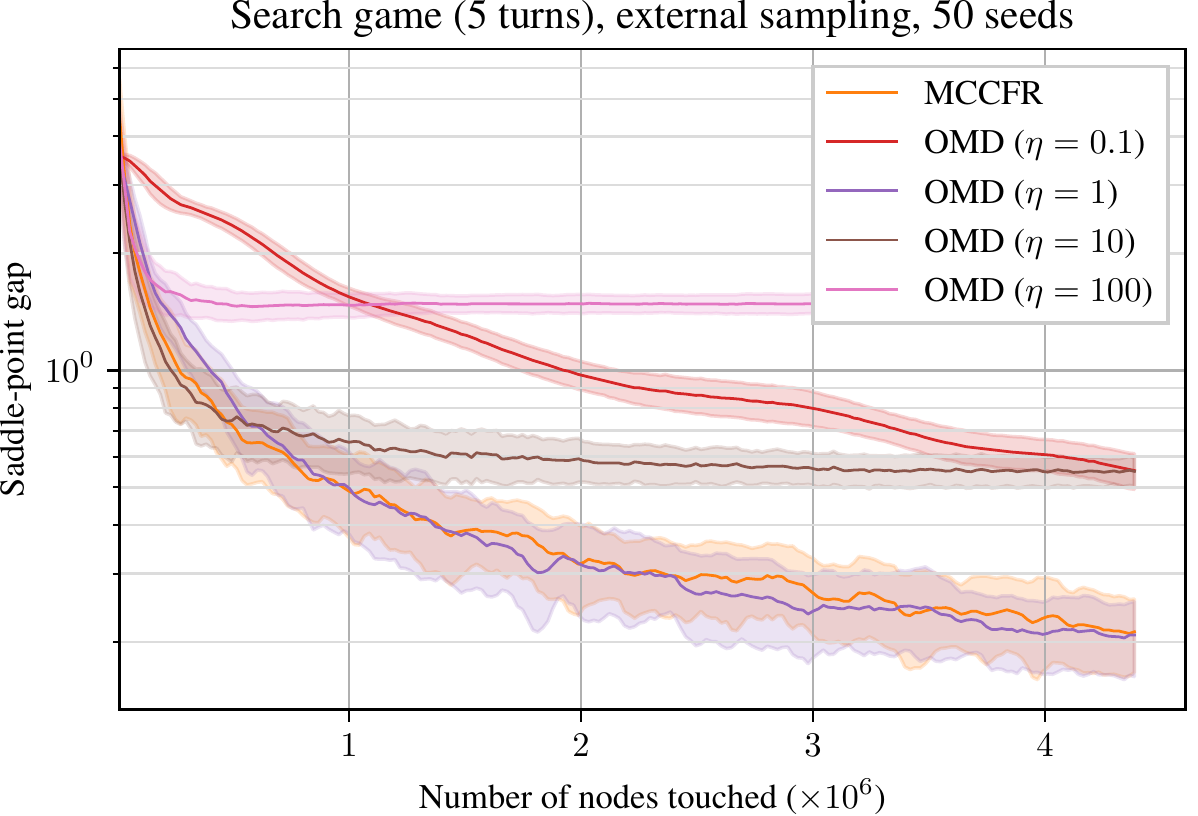}
  \caption{Performance of FTRL and OMD with four stepsizes on Search-5 with external sampling. MCCFR shown for reference}
  \label{fig:appendix search5 external}
\end{figure}

\subsection{Balanced Outcome Sampling}

The Search-4 plot omitted from the main paper is shown here.
\begin{figure}[H]
  \centering
  \includegraphics[width=0.49\columnwidth]{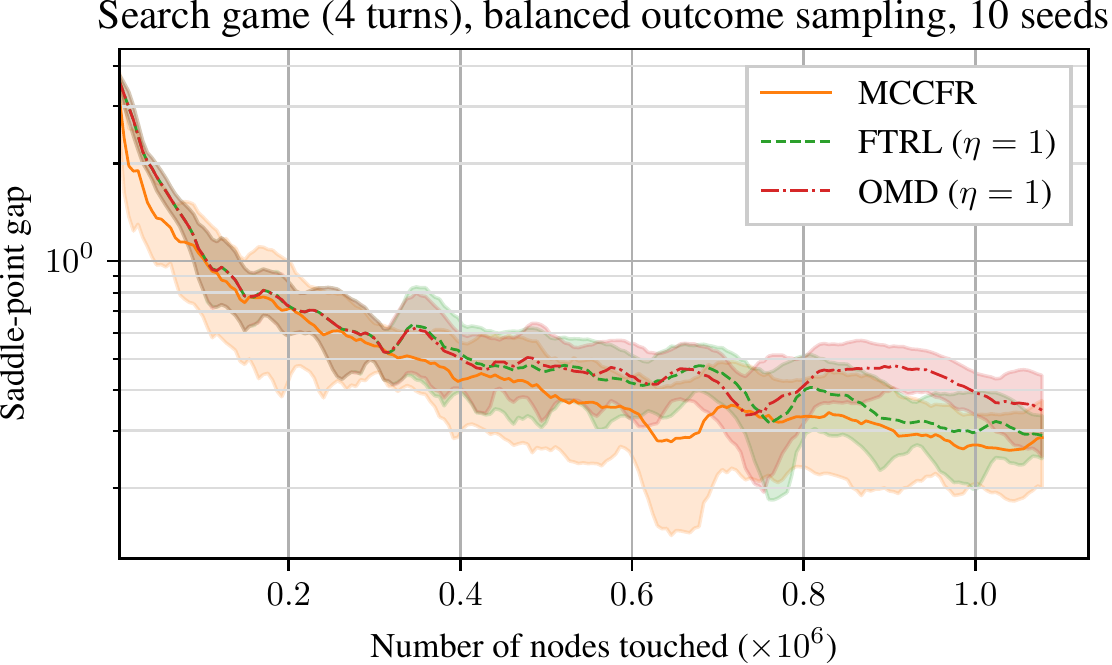}
  \caption{Performance of MCCFR, FTRL, and OMD with outcome sampling on Search-4.}
  \label{fig:search4_outcome}
\end{figure}
Figure~\ref{fig:search4_outcome} shows the performance on Search-4 and Search-5 with outcome sampling. In Search-4 we find that MCCFR performs better than FTRL and OMD, though FTRL is comparable at later iterations.

Figures~\ref{fig:appendix bs outcome} through~\ref{fig:appendix search5 outcome} show the performance of FTRL and OMD with outcome sampling for all four stepsizes that we tried on each game: $\eta=0.1,1,10,100$.

\begin{figure}[H]
  \centering
  \includegraphics[width=0.49\columnwidth]{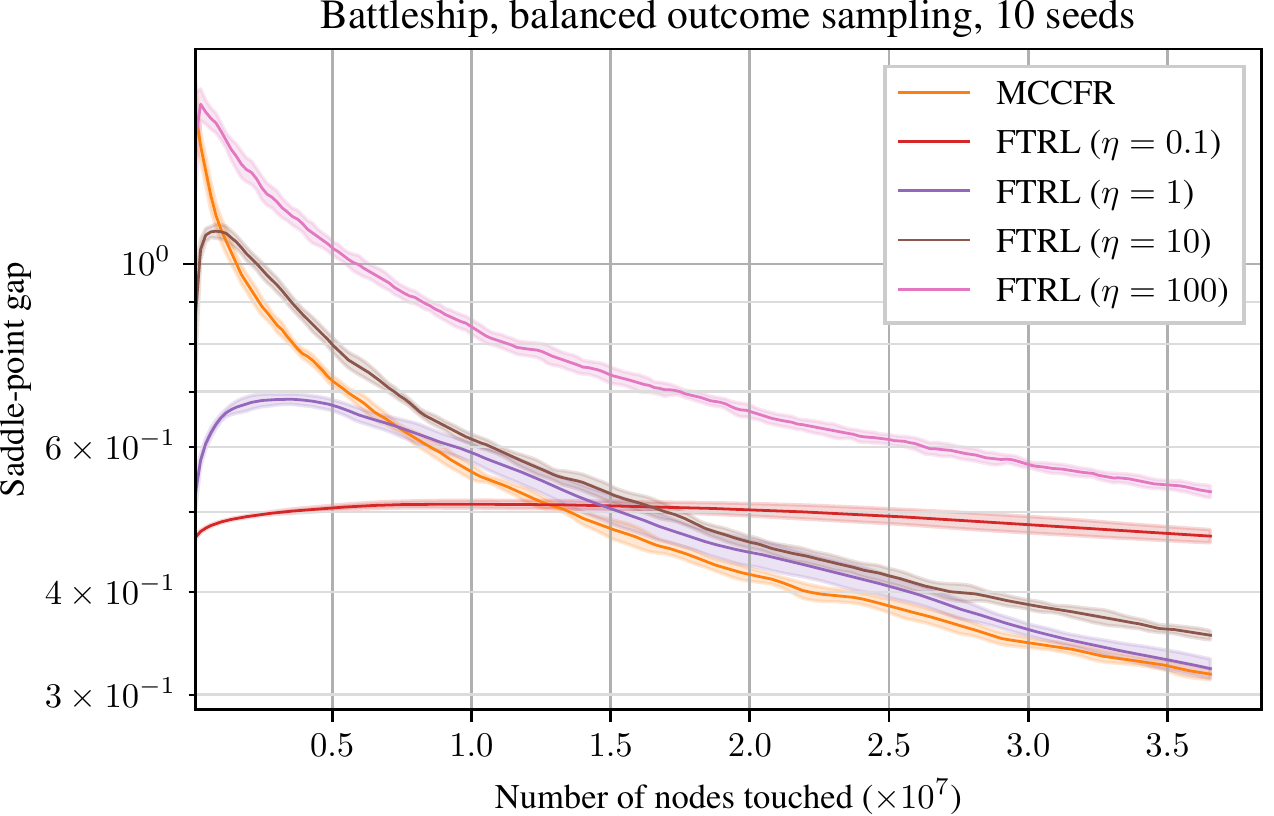}
  \includegraphics[width=0.49\columnwidth]{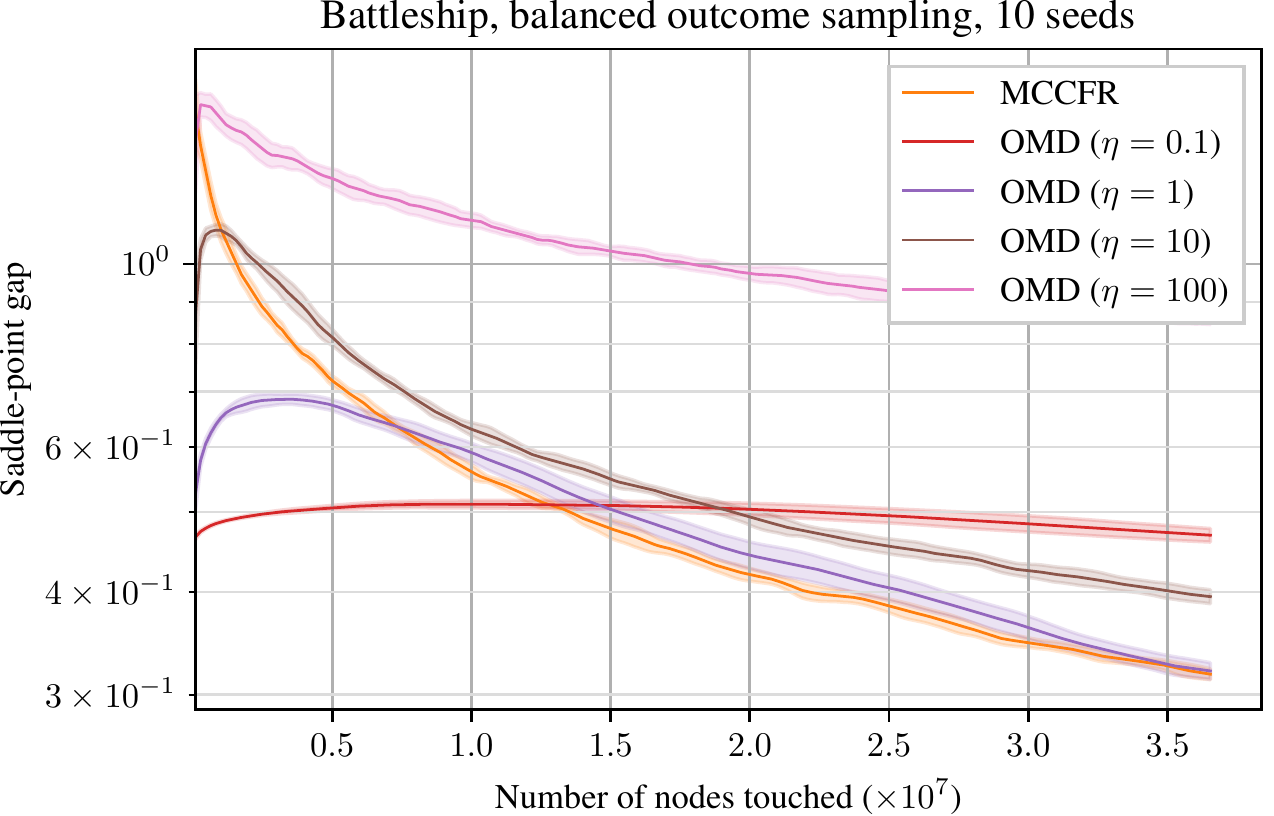}
  \caption{Performance of FTRL and OMD with four stepsizes on Battleship with outcome sampling. MCCFR shown for reference}
  \label{fig:appendix bs outcome}
\end{figure}
\begin{figure}[H]
  \centering
  \includegraphics[width=0.49\columnwidth]{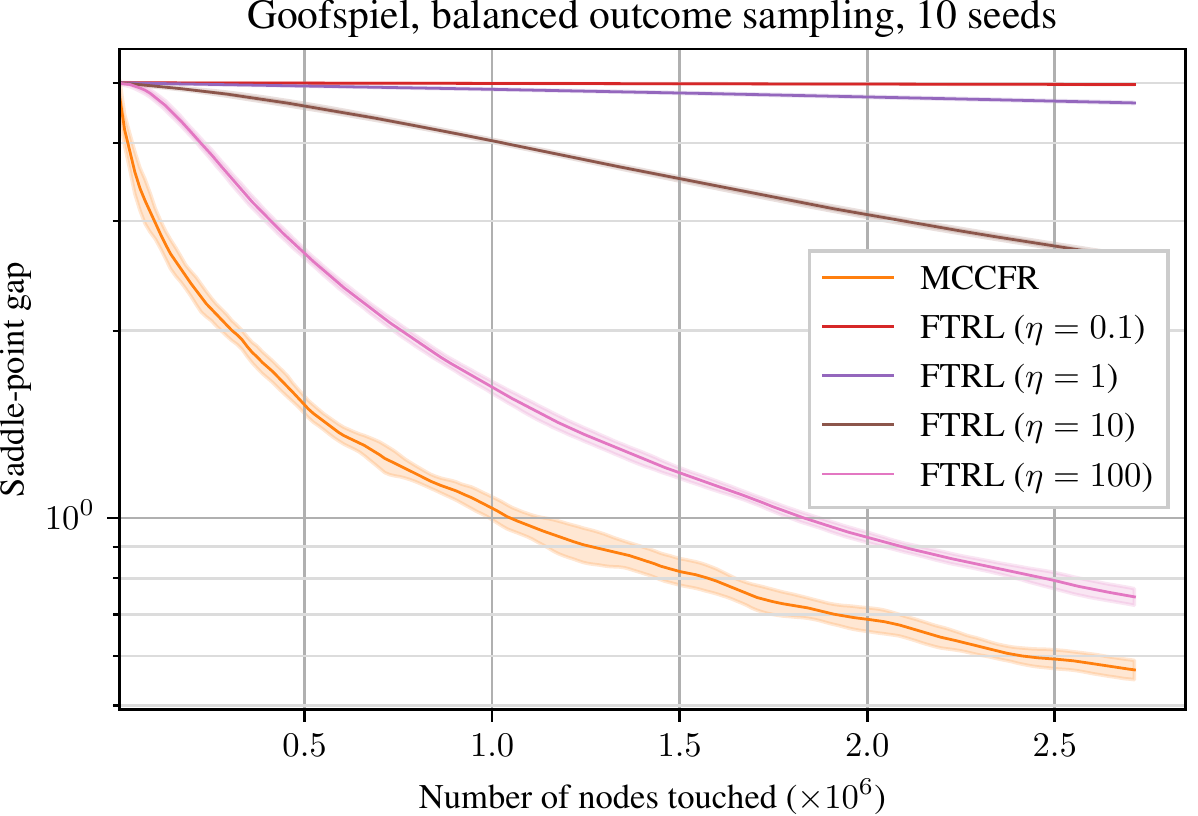}
  \includegraphics[width=0.49\columnwidth]{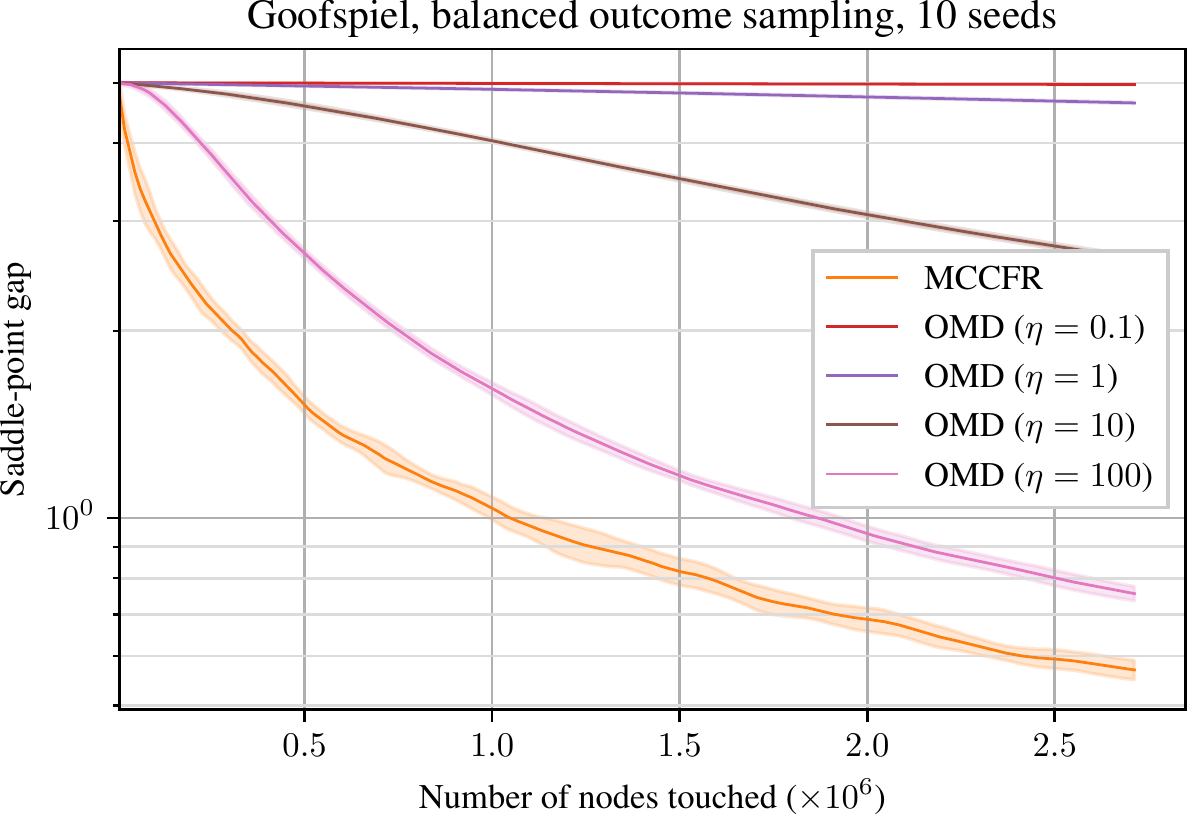}
  \caption{Performance of FTRL and OMD with four stepsizes on Goofspiel with outcome sampling. MCCFR shown for reference}
  \label{fig:appendix goofspiel outcome}
\end{figure}
\begin{figure}[H]
  \centering
  \includegraphics[width=0.49\columnwidth]{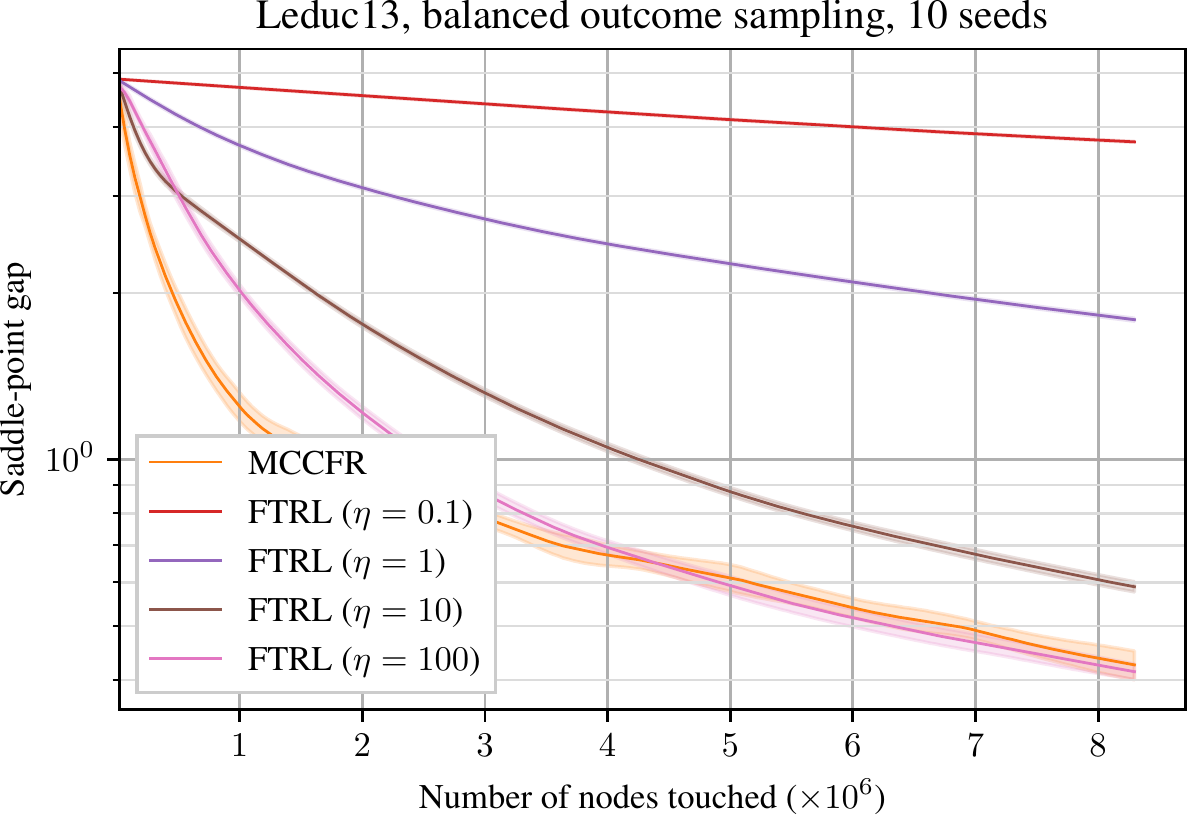}
  \includegraphics[width=0.49\columnwidth]{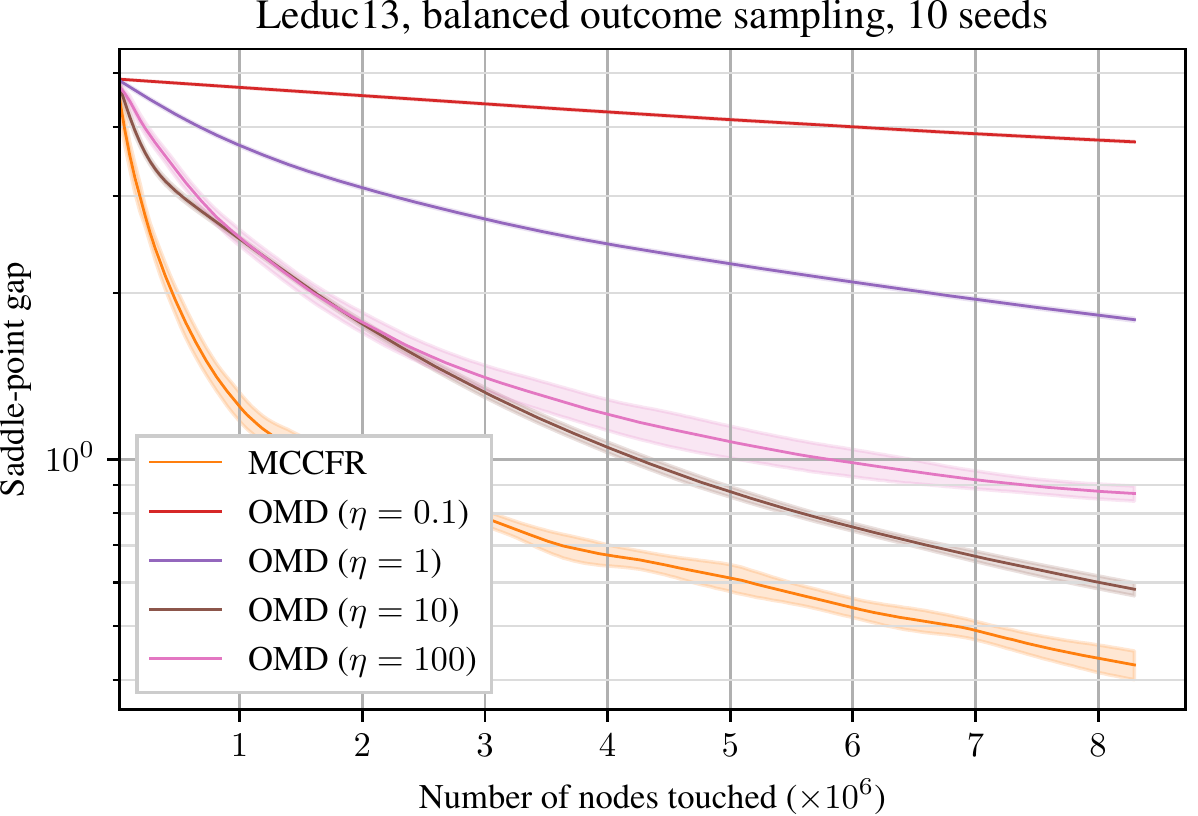}
  \caption{Performance of FTRL and OMD with four stepsizes on Leduc 13 with outcome sampling. MCCFR shown for reference}
  \label{fig:appendix leduc outcome}
\end{figure}
\begin{figure}[H]
  \centering
  \includegraphics[width=0.49\columnwidth]{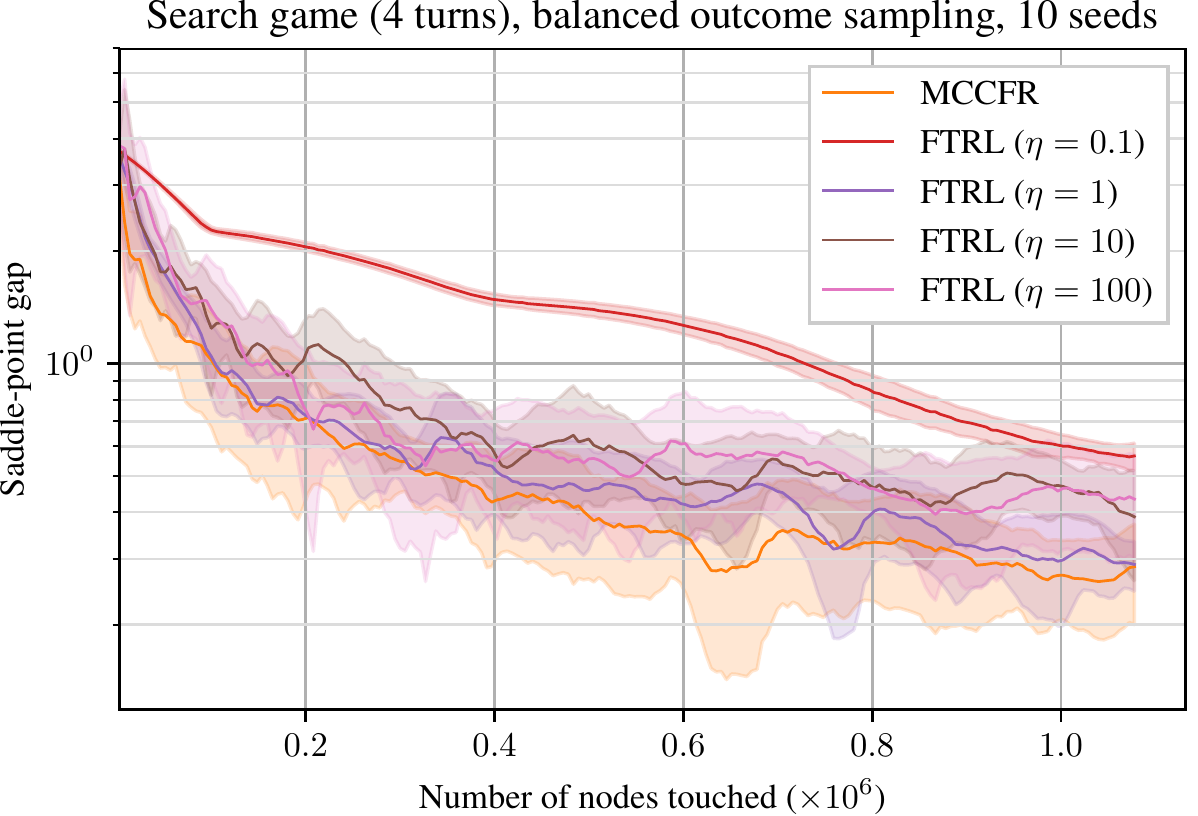}
  \includegraphics[width=0.49\columnwidth]{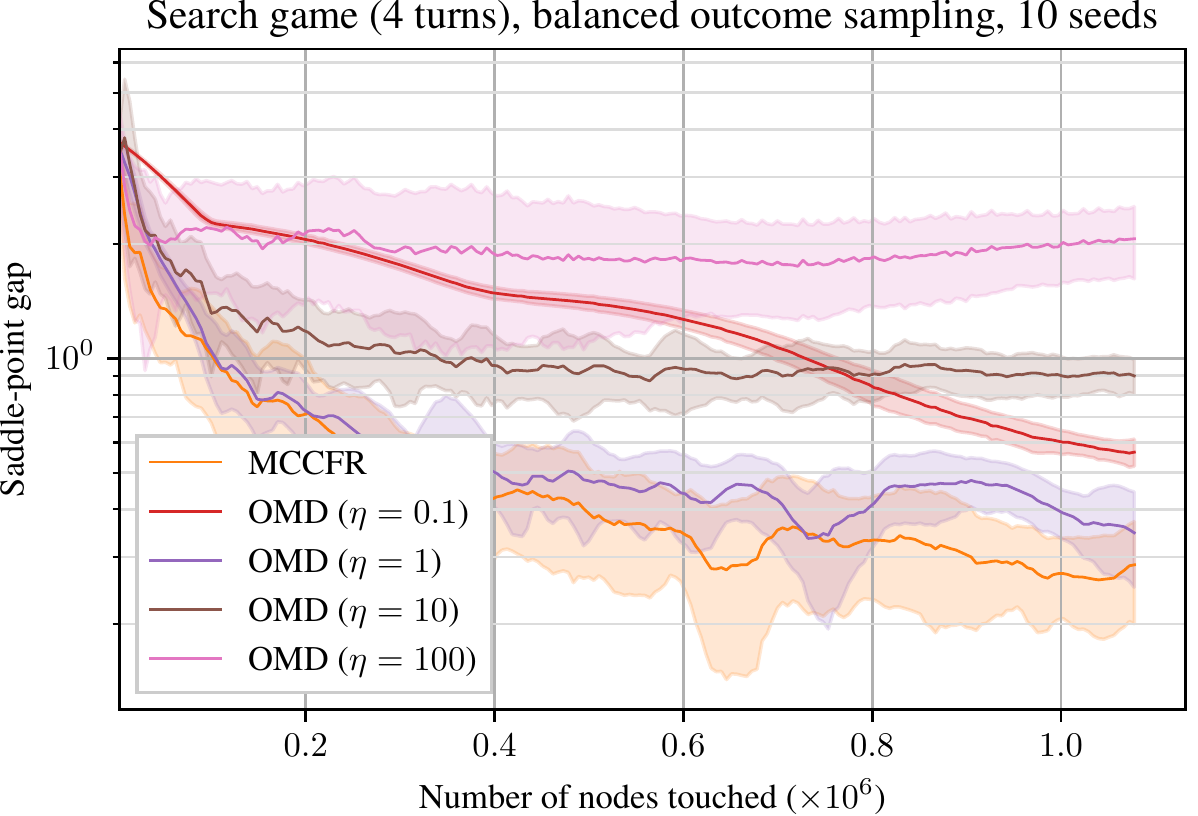}
  \caption{Performance of FTRL and OMD with four stepsizes on Search-4 with outcome sampling. MCCFR shown for reference}
  \label{fig:appendix search4 outcome}
\end{figure}
\begin{figure}[H]
  \centering
  \includegraphics[width=0.49\columnwidth]{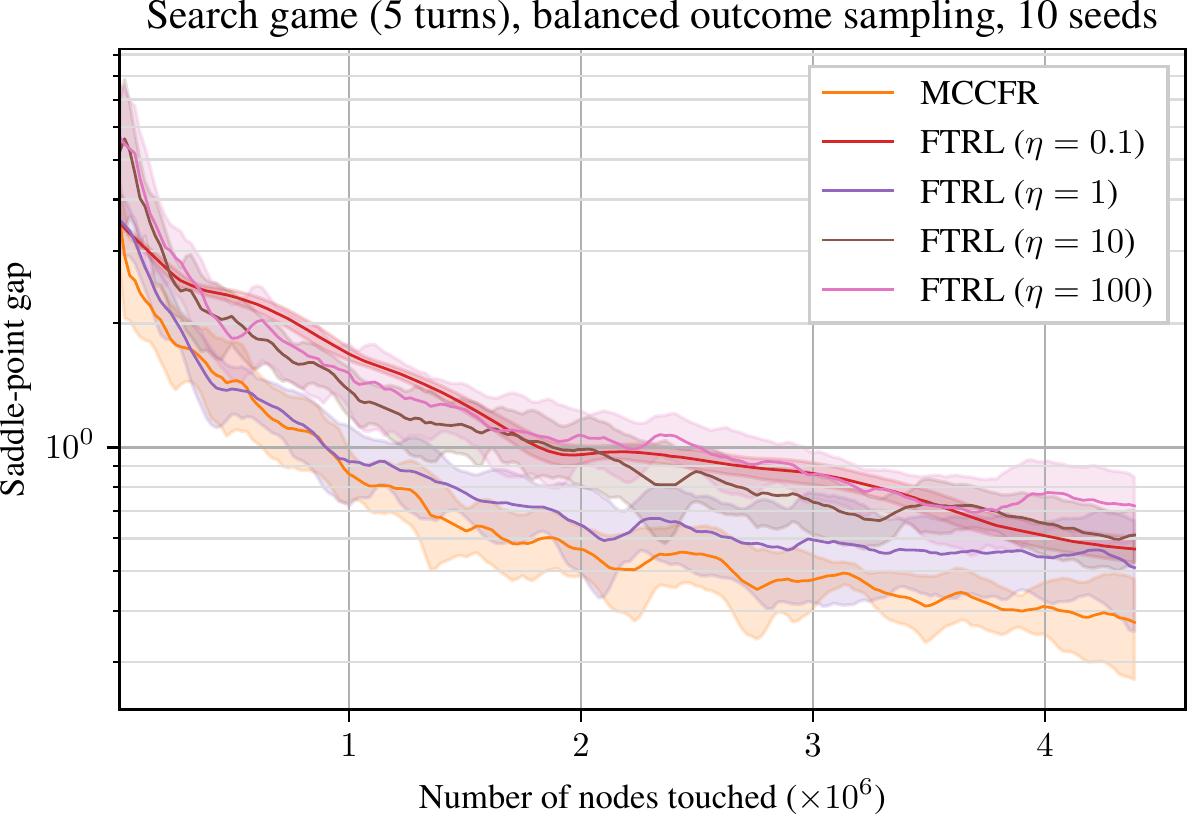}
  \includegraphics[width=0.49\columnwidth]{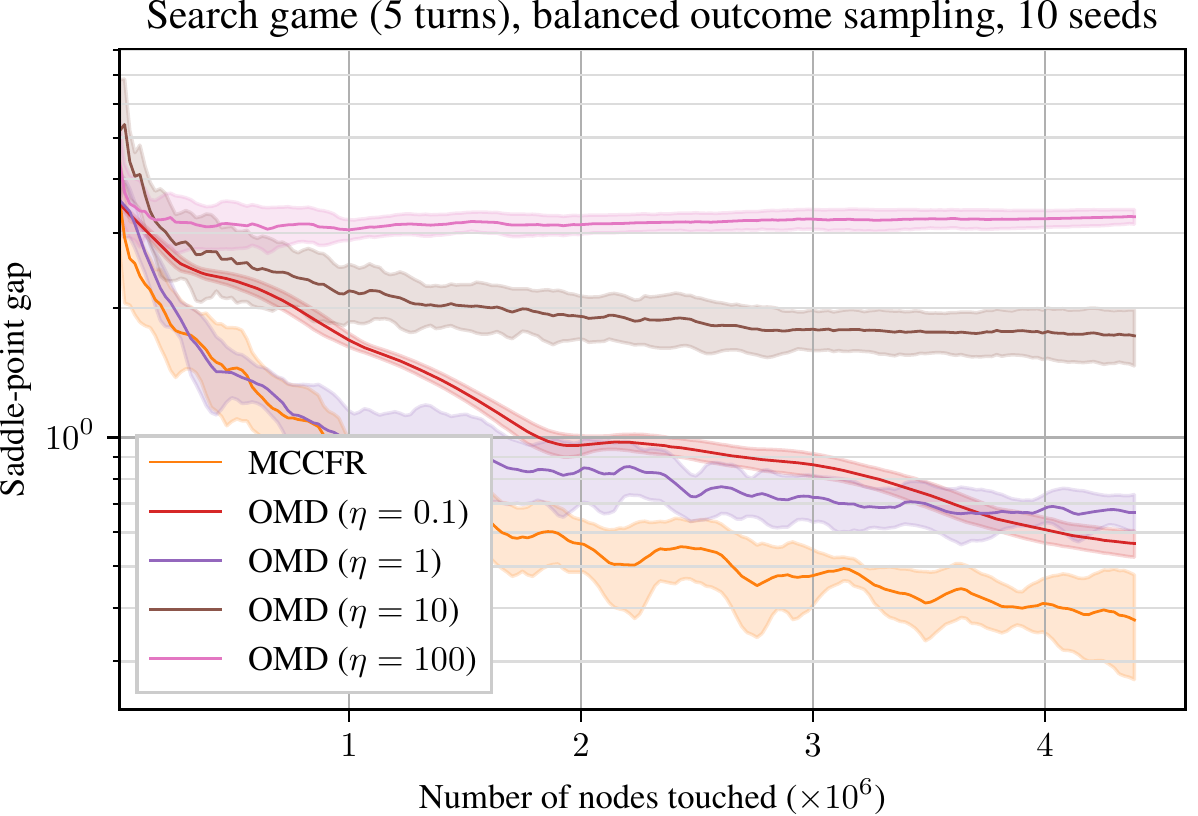}
  \caption{Performance of FTRL and OMD with four stepsizes on Search-5 with outcome sampling. MCCFR shown for reference}
  \label{fig:appendix search5 outcome}
\end{figure}